%% file: paper.tex
\documentclass[runningheads,envcountsame,envcountsect,orivec]{llncs}
\usepackage[T1]{fontenc}
\usepackage{graphicx}
\usepackage{color}

\usepackage{bbding} 
\usepackage{amssymb}
\usepackage{thm-restate}
\usepackage{ifthen}
\usepackage{url}

\spnewtheorem*{convention}{Convention}{\itshape}{}

\def\cA{\mathcal{A}}
\def\cC{\mathcal{C}}
\def\cD{\mathcal{D}}
\def\cG{\mathcal{G}}
\def\cI{\mathcal{I}}
\def\cJ{\mathcal{J}}
\def\cR{\mathcal{R}}
\def\cS{\mathcal{S}}
\def\cU{\mathcal{U}}
\def\cV{\mathcal{V}}
\newcommand{\sort}[2][]{\ifthenelse{\equal{#1}{}}{\mbox{\sl{#2}}}{\mbox{\sl#1{#2}}}}
\newcommand{\SortBool}[1][]{\sort[#1]{bool}}
\def\var#1{\mathit{#1}}
\newcommand{\symb}[1]{\mathsf{#1}}
\newcommand{\Pos}{\mathcal{P}os}
\newcommand{\Var}{\mathcal{V}ar}
\newcommand{\Root}{\mathit{root}}
\newcommand{\Dom}{\mathcal{D}om}
\newcommand{\Ran}{\mathcal{R}an}
\newcommand{\VRan}{\mathcal{VR}an}
\newcommand{\mgu}[1][]{\ifthenelse{\equal{#1}{}}{\mathit{mgu}}{\mathit{mgu}_{#1}}}
\newcommand{\cStheory}[1][\mathit{theory}]{\ifthenelse{\equal{#1}{}}{\cS}{\cS_{#1}}}
\newcommand{\cScore}{\cStheory[\mathit{core}]}

\newcommand{\cSterm}{\cS_{\mathit{term}}}
\newcommand{\Sigtheory}[1][\mathit{theory}]{\ifthenelse{\equal{#1}{}}{\Sigma}{\Sigma_{#1}}}
\newcommand{\Sigcore}{\Sigtheory[\mathit{core}]}
\newcommand{\Sigint}{\Sigtheory[{\sort[\scriptsize]{int}}]}
\newcommand{\Sigterm}{\Sigma_{\mathit{terms}}}
\newcommand{\Val}[1][]{%
\ifthenelse{\equal{#1}{}}{%
{\mathcal{V}al}}{{\mathcal{V}al}_{#1}}}
\newcommand{\Eval}[2][]{\ifthenelse{\equal{#1}{}}{{[\![ #2 ]\!]}}{{[\![ #2 ]\!]}_{#1}}}
\newcommand{\Constraint}[1]{[ #1 ]}
\newcommand{\CTerm}[2]{#1 \, \Constraint{#2}}
\newcommand{\LVar}{\mathcal{LV}ar}
\newcommand{\cRcalc}{\cR_{\mathit{calc}}}
\newcommand{\NF}{\mathit{NF}}
\newcommand{\Height}[1]{\mathit{height}(#1)}
\newcommand{\GInst}[2][]{\ifthenelse{\equal{#1}{}}{\cG{G}(#2)}{\cG_{#1}(#2)}}
\newcommand{\Cocterm}[2][]{\overline{#2}}
\newcommand{\Cosubst}[2][]{\overline{#2}}
\newcommand{\Copattern}[3][]{\overline{(#2,#3)}}
\newcommand{\CopatternSet}[2][]{\ifthenelse{\equal{#1}{}}{\overline{#2}}{\overline{(#2)}_{#1}}}
\newcommand{\SuccsimH}[1][h]{\succsim_{\leq #1}}
\newcommand{\SuccH}[1][h]{\succ_{\leq #1}}
\newcommand{\PrecsimH}[1][h]{\precsim_{\leq #1}}

\newcommand{\SuccsimHN}[1][h]{\succsim_{\leq #1,\mathbb{N}}}
\newcommand{\SuccHN}[1][h]{\succ_{\leq #1,\mathbb{N}}}
\newcommand{\PrecsimHN}[1][h]{\precsim_{\leq #1,\mathbb{N}}}

\newcommand{\UnifStep}[1][]{\mathrel{\ifthenelse{\equal{#1}{}}{\Rrightarrow}{\Rrightarrow_{\mathsf{#1}}}}}
\def\doteq{\mathrel{\dot{=}}}
\def\dotcup{\mathop{\dot{\cup}}}
\def\dotuplus{\mathop{\dot{\uplus}}}
\newcommand{\Cdot}{\,\rule{0pt}{6pt}\cdot\,}
\def\cSzlist{\cS_1}
\def\cCzlist{\cC_1}
\def\Sigmaflist{\Sigma_1}
\def\cDflist{\cD_1}
\def\cRflist{\cR_1}

\begin{document}
\title{%
Difference of Constrained Patterns in Logically Constrained Term Rewrite Systems%
\thanks{This work was partially supported by 
JSPS KAKENHI Grant Number JP24K02900, 
and
Grant-in-Aid for JSPS Fellows Grant Number JP24KJ1240.
}
}
\subtitle{(Full Version)}
\titlerunning{Difference of Constrained Patterns in LCTRSs (Full Version)}
%
\author{%
Naoki Nishida
\Envelope%
\and
Misaki Kojima
\Envelope%
\and
Yuto Nakamura
}
\authorrunning{N. Nishida, K. Misaki, and Y. Nakamura}
%
\institute{%
Graduate School of Informatics,
Nagoya University, \\
Furo-cho, Chikusa-ku, 
4648601 
Nagoya, Japan\\
\email{nishida@i.nagoya-u.ac.jp}
\quad
\email{kojima@i.nagoya-u.ac.jp}
}
\maketitle              
\begin{abstract}
Considering patterns as sets of their instances, a difference operator over patterns computes a finite set of two given patterns, which represents the difference between the dividend pattern and the divisor pattern.
A complement of a pattern is a pattern set, the ground constructor instances of which comprise the complement of the ground constructor instances of the former pattern.
Given finitely many unconstrained linear patterns, using a difference operator over linear patterns, a complement algorithm returns a finite set of linear patterns as a complement of the given patterns.
In this paper, we extend the difference operator and complement algorithm to constrained linear patterns used in 
logically constrained term rewrite systems (LCTRSs, for short) that have no user-defined constructor term with a sort for built-in values.
Then, as for left-linear term rewrite systems, using the complement algorithm, we show that quasi-reducibility is decidable for such LCTRSs with decidable built-in theories.
For the single use of the difference operator over constrained patterns, only divisor patterns are required to be linear.

\end{abstract}

\section{Introduction}
\label{sec:intro}

\emph{Complements} of \emph{patterns}---terms rooted by defined symbols with constructor arguments---have been studied in the field of term rewriting due to their various usefulness~\cite{Taj93,LM86,LLT90,Fer98,HA19pro}.
A \emph{complement} of a pattern $f(s_1,\ldots, s_n)$
is a set $Q$ of patterns, the ground constructor instances of which comprise the complement set of the ground constructor instances of $f(s_1,\ldots,s_n)$:
$ 
\GInst[\cC]{Q}
=
\GInst[\cC]{f(x_1,\ldots,x_n)}
\setminus
\GInst[\cC]{f(s_1,\ldots,s_n)}
$, 
where $x_1,\ldots,x_n$ are pairwise distinct variables, $\cC$ is a set of constructors, $T(\cC)$ is the set of ground constructor terms, 
$\GInst[\cC]{s}$ is the set of ground constructor instances of a term $s$,
and 
$\GInst[\cC]{Q} = \bigcup_{s \in Q} \GInst[\cC]{s}$.

Given a finite set $P$ of linear patterns, the \emph{complement algorithm} in~\cite{LLT90,HA19pro} returns a finite set $\CopatternSet{P}$ of linear patterns as a complement of patterns in $P$:
$
\GInst[\cC]{\CopatternSet{P}}
=
\GInst[\cC]{\{f(x_1,\ldots,x_n) \mid f \in \cD\}}
\setminus
\GInst[\cC]{P}
$,
where $\cD$ is a set of considered defined symbols and $x_1,\ldots,x_n$ are pairwise distinct variables.
Patterns in $P$ are assumed to be pairwise non-overlapping.
The algorithm is based on a \emph{difference} operator $\ominus$ over linear patterns.
The difference operator $\ominus$ takes two linear patterns $s,t$ as input and returns a finite set of linear patterns, which represents the difference of the \emph{dividend} $s$ and the \emph{divisor} $t$ w.r.t.\ ground constructor instances: 
$ 
\GInst[\cC]{s \ominus t} 
= \GInst[\cC]{s} \setminus \GInst[\cC]{t}
$. 

The operator $\ominus$ is extended for finite sets of linear patterns:
Given finite sets $P,Q$ of linear patterns, $P \ominus Q$ is a finite set of linear patterns such that 
$
\GInst[\cC]{P \ominus Q}
= 
\GInst[\cC]{P}
\setminus 
\GInst[\cC]{Q}
$.
As an application, the complement algorithm implies decidability of \emph{quasi-reducibility}%
\footnote{Quasi-reducibility is called \emph{quasi-reductivity} in some papers (e.g.,~\cite{FKN17tocl,Kop17,ATK17})
and \emph{pattern completeness} in the context of functional programming~\cite{TY24}.}%
---non-existence of irreducible ground patterns---of left-linear TRSs:
A left-linear TRS $\cR$ is quasi-reducible if and only if $\{ f(x_1,\ldots,x_n) \mid 
f \in \cD
\} \ominus \{ \ell \mid \ell \to r \in \cR \} = \emptyset$, 
where $x_1,\ldots,x_n$ are pairwise different variables.
Quasi-reducibility is often assumed for target rewrite systems in equivalence checking based on \emph{rewriting induction}~\cite{Red90,FKN17tocl}.

Many compilers and interpreters of practical programming languages check the exhaustiveness of function definitions and case-statements.
As for left-linear TRSs, algorithms for the exhaustiveness checking have been developed in several formulations (see, e.g,~\cite{Mar07}).
For patterns without guard conditions, the exhaustiveness (i.e., quasi-reducibility) is decidable~\cite{KNZ87,TY24}.
On the other hand, to the best of our knowledge, there is no full exhaustiveness checker for languages that allows us to attach guard conditions, which may include user-defined predicates, to patterns.
To decide it, we need decidable SMT solving for the combination of all built-in theories that the languages support, which are usually undecidable.

Recently, program verification by means of \emph{logically constrained term rewrite systems} (LCTRSs, for short)~\cite{KN13frocos} are well investigated~\cite{FKN17tocl,WM18,CL18,NW18vstte,KN18eptcs,KNS19ss,KNM20wpte,KN23padl,KN23jlamp,KNM25jlamp,NKM25jlamp}.
LCTRSs are extensions of \emph{term rewrite systems} (TRSs, for short) by allowing rewrite rules to have guard constraints which are evaluated under equipped built-in theories.
LCTRSs combine classic \emph{term rewriting} (see, e.g.,~\cite{BN98,Ohl02}) with
built-in data types and constraints from user-specified first-order theories,
specifically those supported by modern 
SMT 
solvers (cf.~\cite{BFT16,BSST21}). %
This allows for a high expressivity that is useful for representing many programming
language constructs directly, together with robust tool
support, e.g., the tool \textsf{Ctrl}~\cite{KN15lpar}, for automated reasoning.
For instance, equivalence checking by means of LCTRSs is useful to ensure the correctness of terminating functions (cf.~\cite{FKN17tocl}). 
%
Due to these features, LCTRSs are known to be useful computational models of not only functional but also imperative programs.


LCTRS tools~\cite{KN15lpar,KNM25jlamp,SMM24} rely on SMT solvers and often consider decidable built-in theories for given LCTRSs. 
On the other hand, to the best of our knowledge, neither a difference operator nor a complement algorithm has been investigated for constrained patterns, while an SMT-based sufficient condition for left-patternless constrained rewrite systems~\cite{SNSSK09,KNM25jlamp}%
\footnote{The sufficient condition for \emph{reduction-completeness} is applicable to quasi-reducibility.}
and procedural sufficient conditions for quasi-reducibility of LCTRSs~\cite{Kop17}
and conditional constrained TRSs~\cite{BJ12} have been shown.

\begin{example}
\label{ex:flist}
Let us consider the sort set $\cSzlist = \{ \sort{int}, \sort{bool}, \sort{list} \}$, the  $\cSzlist$-sorted signature
$
\Sigmaflist = 
\{ \symb{n}: \sort{int} \mid n \in \mathbb{Z} \}
\cup
\{
\symb{true},\symb{false}: \sort{bool}, \,
\symb{nil}: \sort{list}, \,
\symb{cons}: \sort{int} \times \sort{list} \to \sort{list}, \,
\symb{f}: \sort{list} \times \sort{int} \to \sort{int}
\}
$, and the following artificial LCTRS:
\[
\cRflist {=}
\left\{
\begin{array}{@{}c@{~}r@{\,}c@{\,}l@{\,}l@{\!}}
(1) 
& \symb{f}(\symb{nil}, y_1) & \to & \symb{0} & \Constraint{\,y_1 \leq \symb{0}\,}, \\
(2)
& \symb{f}(\symb{cons}(x_2,\var{xs}_2), y_2) & \to & \symb{f}(\var{xs}_2, y_2 - \symb{1}) & \Constraint{\, x_2 \leq \symb{0} \land y_2 > \symb{0} \,}, \\
(3)
& \symb{f}(\symb{cons}(x_3,\symb{cons}(z_3,\var{zs}_3)), y_3) & \to & x_3 + \symb{f}(\var{zs}_3, y_3 - \symb{2}) & \Constraint{\, x_3 > \symb{0} \land y_3 > \symb{1} \,} \\
\end{array}
\right\}
\]
Complements of unconstrained patterns $\symb{f}(\symb{nil},y_1)$ and $\symb{f}(\symb{cons}(x_1,\symb{nil}),y_1)$ are, e.g.,
$\{ \symb{f}(\symb{cons}(x_1,\var{xs}_1),y_1) \}$
and
$\{ \symb{f}(\symb{nil},y_1), \symb{f}(\symb{cons}(x_1,\symb{cons}(z_1,\var{zs}_1)),y_1) \}$, respectively.
The LCTRS $\cRflist$ is not quasi-reducible because none of the following constrained patterns is defined: 
\[
\begin{array}{@{}c@{~}r@{~}l@{}}
(\mathrm{a}) & \symb{f}(\symb{nil},y_{\mathrm{a}}) & \Constraint{\, \neg (y_{\mathrm{a}} \leq \symb{0}) \,} \\
(\mathrm{b}) & \symb{f}(\symb{cons}(x_{\mathrm{b}},\symb{nil}),y_{\mathrm{b}}) & \Constraint{\, \neg(x_{\mathrm{b}} \leq \symb{0} \land y_{\mathrm{b}} > \symb{0}) \,} \\
(\mathrm{c}) & \symb{f}(\symb{cons}(x_{\mathrm{c}},\symb{cons}(z_{\mathrm{c}},\var{zs}_{\mathrm{c}})),y_{\mathrm{c}}) 
& \Constraint{\, 
\neg(x_{\mathrm{c}} \leq \symb{0} \land y_{\mathrm{c}} > \symb{0})
\land
\neg(x_{\mathrm{c}} > \symb{0} \land y_{\mathrm{c}} > \symb{1})
\,} \\
\end{array}
\]
It would be easy to find the first undefined constrained pattern~(a) above.
On the other hand, 
the second and third ones~(b),\,(c) above are not so trivial because we have to consider a unified form of the left-hand sides of the second and third rules~(2),\,(3) to compute~(b),\,(c).
For this reason, it is not so easy to know that $\cRflist$ is not quasi-reducible.
Note that the following LCTRS obtained from $\cRflist$ by adding rules for the above undefined constrained patterns is quasi-reducible:
\[
\cRflist' = 
\cRflist 
\cup
\left\{
\begin{array}{@{}c@{~}r@{\>}c@{\>}l@{~}l@{}}
(\mathrm{a})
& \symb{f}(\symb{nil},y_{\mathrm{a}}) & \to & \symb{0} & \Constraint{\, \neg(y_{\mathrm{a}} \leq \symb{0}) \,}, \\
(\mathrm{b})
& \symb{f}(\symb{cons}(x_{\mathrm{b}},\symb{nil}),y_{\mathrm{b}}) & \to & \symb{0} & \Constraint{\, \neg(x_{\mathrm{b}} \leq \symb{0} \land y_{\mathrm{b}} > \symb{0}) \,}, \\
(\mathrm{c})
& \symb{f}(\symb{cons}(x_{\mathrm{c}},\symb{cons}(z_{\mathrm{c}},\var{zs}_{\mathrm{c}})),y_{\mathrm{c}}) & \to & \symb{0}
& \Constraint{\, 
\neg(x_{\mathrm{c}} \leq \symb{0} \land y_{\mathrm{c}} > \symb{0})
\\\
&&&& ~{} 
\land
\neg(x_{\mathrm{c}} > \symb{0} \land y_{\mathrm{c}} > \symb{1})
\,} \\
\end{array}
\right\}
\]
\end{example}

In this paper, we propose a difference operator and a complement algorithm for constrained patterns used in left-linear LCTRSs that have finitely many user-defined function symbols and have no user-defined constructor term with a theory sort for built-in values.%
\footnote{The latter is equivalent to the condition that all theory sorts are \emph{inextensible}~\cite{FGK25}.}
To this end, we extend the difference operator $\ominus$
and the complement algorithm $\CopatternSet{\Cdot}$ mentioned above
to constrained patterns. 
The extended difference operator and complement algorithm are still computable for constrained patterns over signatures with decidable built-in theories. 
Then, as an application of the complement algorithm, we show that quasi-reducibility is decidable for the aforementioned LCTRSs with decidable built-in theories.
Complete proofs for correctness of the results can be seen in the appendix. 

\begin{example}[our goal]
\label{ex:goal}
Let us consider 
the LCTRSs $\cRflist$ and $\cRflist'$ in Example~\ref{ex:flist} again.
Given the constrained patterns $\CTerm{\symb{f}(x,\var{ys})}{\symb{true}}$ and~(1), the extended difference operator $\ominus$ returns
$\{\, 
\CTerm{\symb{f}(\symb{nil},y_{\mathrm{a}})}{\, \neg(y_{\mathrm{a}} \leq \symb{0}) \,}
, \, \CTerm{\symb{f}(\symb{cons}(x,\var{xs}),y)}{\,\symb{true}\,} \,\}$ as a complement of~(1). 
Given $\{(1),(2),(3)\}$, the extended complement algorithm $\CopatternSet{\Cdot}$ returns $\{(\mathrm{a}),(\mathrm{b}),(\mathrm{c})\}$ as a complement of $\{(1),(2),(3)\}$. 
This implies that $\cRflist$ is not quasi-reducible.
On the other hand, given $\{(1),(2),(3),(\mathrm{a}),(\mathrm{b}),(\mathrm{c})\}$, the complement algorithm $\CopatternSet{\Cdot}$ returns the empty set, and thus $\cRflist'$ is quasi-reducible.
\end{example}

In addition to the extension to constrained patterns, we relax the linearity assumption of dividends of $\ominus$ over constrained patterns, i.e., we no longer assume the linearity of dividends, while the linearity is assumed for the extended complement algorithm:
For $\CTerm{s}{\phi} \ominus \CTerm{t}{\psi}$, only the dvisor $t$ is assumed to be linear;
for $P \ominus Q$, both $P,Q$ are assumed to be sets of constrained linear patterns.
%
The extended operator and complement algorithm behave in the same way as those of unconstrained patterns.
Thus, a side effect of the extension to constrained patterns, 
for the unconstrained case with $s \ominus t$, only $t$ is assumed to be linear.


LCTRSs in this paper are assumed to have no user-defined constructor term with a theory sort.
This is not a restriction for program verification by means of LCTRSs because LCTRSs obtained from practical programs by the existing transformations~\cite{FKN17tocl,KN18eptcs,KNS19ss,KNM25jlamp} are left-linear 
systems satisfying the restriction.

The main contributions of this paper are 
(i) 
the extension of the difference operator and the complement algorithm 
to constrained patterns,
(ii)
the relaxation of the linearity assumption to dividends of $\ominus$ over (constrained) patterns,
and
(iii)
decidability of quasi-reducibility of LCTRSs in the aforementioned class.


\section{Preliminaries}
\label{sec:preliminaries}

In this section, we briefly recall LCTRSs~\cite{KN13frocos,FKN17tocl}.
Familiarity with basic notions and notations on term rewriting~\cite{BN98,Ohl02} is assumed. 

This paper deals with a first-order $\cS$-sorted signature $\Sigma$, where $\cS$ is a set of sorts.
Let $\Sigma' \subseteq \Sigma$.
A term in $T(\Sigma',\cV)$ is called a \emph{$\Sigma'$-term}, where $\cV$ is a (countably infinite) set of $\cS$-sorted variables.
A term $t$ is called an \emph{instance} of a term $t'$ (or $t'$ is called \emph{more general than $t$}), written as $t \gtrsim t'$, if there exists a substitution $\sigma$ such that $t = t'\sigma$.
We write $t > t'$ if $t \gtrsim t'$ but $t' \not\gtrsim t$.
The \emph{height} of a term $t$, written as $\Height{t}$, is recursively defined as follows:
$\Height{t} = 0$ if $t \in \cV$;
$\Height{t} = 1 + \max \{ \Height{t_1},\ldots, \Height{t_n} \}$ if $t = f(t_1,\ldots,t_n)$ for some $n$-ary function symbol $f$ and terms $t_1,\ldots,t_n$ such that $n \geq 0$.
Note that the height of $t$ in this paper does not correspond to the height of the corresponding tree of $t$.
A position $p$ of a term $t$ is called \emph{above} (\emph{strictly above}) a position $q$ of $t$, written as $p \leq q$ ($p < q$), if there exists a position $p'$ of $t|_p$ such that $pp' = q$ ($pp' = q$ and $p'\ne\epsilon$).
Positions $p,q$ of a term $t$ are called \emph{parallel}, written as $p \mathrel{||} q$, if $p \not\leq q$ and $q \not \leq p$.
A substitution $\sigma$ is called \emph{more general than} a substitution $\theta$, written as $\sigma \lesssim \theta$, if there exists a substitution $\delta$ such that $\theta = (\delta \circ \sigma)$.
For a set $X$ of variables, we call $\sigma$ \emph{ground w.r.t.\ $X$} ($X$-ground, for short) if $\Dom(\sigma) \supseteq X$ and $\VRan(\sigma|_X) = \emptyset$, i.e., $x\sigma$ is ground for any variable $x \in X$. 
A substitution $\sigma$ which is a sort-preserving mapping from $\cV$ to $T(\Sigma,\cV)$ is said to be a \emph{$\Sigma'$-substitution} if $\Ran(\sigma) \subseteq T(\Sigma',\cV)$.
A term $s$ is called an \emph{instance of $t$ w.r.t.\ $\Sigma'$} ($\Sigma'$-instance, for short) if there exists a $\Sigma'$-substitution $\sigma$ such that $s = t\sigma$.
For a substitution $\sigma$, we write $\sigma = \mgu(s,t)$ if $\sigma$ is an mgu of $s,t$.


To define an LCTRS~\cite{KN13frocos,FKN17tocl}  over an $\cS$-sorted signature $\Sigma$, we consider the following sorts, signatures, mappings, and constants:
\emph{Theory sorts} in $\cStheory$ and \emph{term sorts} in $\cSterm$ such that $\cS=\cStheory \uplus \cSterm$;
a \emph{theory signature} $\Sigtheory$ and a \emph{term signature} $\Sigterm$ such that $\Sigma = \Sigtheory \cup \Sigterm$ and 
$\iota_1,\ldots,\iota_n,\iota\in \cStheory$ for any symbol $f: \iota_1 \times \cdots \times \iota_n \to \iota \in \Sigtheory$;
a mapping $\cI$ that assigns to each theory sort $\iota$ a (non-empty) set $\cA_\iota$, called the universe of $\iota$ (i.e., $\cI(\iota) = \cA_\iota$);
a mapping $\cJ$, called an interpretation for $\Sigtheory$, that assigns to each function symbol $f : \iota_1 \times \cdots \times \iota_n \to \iota \in \Sigtheory$ a function $f^\cJ$ in 
$\cI(\iota_1) \times \cdots \times \cI(\iota_n) \to \cI(\iota)$ (i.e., $\cJ(f)=f^\cJ$);
a set $\Val[\iota] \subseteq \Sigtheory$ of \emph{value-constants} $a: \iota$ for each theory sort $\iota$ 
such that $\cJ$ gives a bijection from $\Val[\iota]$ to $\cI(\iota)$.
We denote $\bigcup_{\iota \in \cStheory} \Val[\iota]$ by $\Val$.
Note that $\Val \subseteq \Sigtheory$.
For readability, we may not distinguish $\Val[\iota]$ and $\cI(\iota)$, i.e., for each $v \in \Val[\iota]$, $v$ and $\cJ(v)$ may be identified.
We require that $\Sigterm \cap \Sigtheory \subseteq \Val$.
Symbols in ${\Sigtheory} \setminus {\Val}$ are \emph{calculation symbols}, for which we may use infix notation.
A term in $T(\Sigtheory,\cV)$ is called a \emph{theory term}. 
We recursively define the \emph{interpretation} $\Eval[\cJ]{\,\cdot\,}$ of ground theory terms as $\Eval[\cJ]{ f(s_1,\ldots,s_n) } = \cJ(f)(\Eval[\cJ]{ s_1 },\ldots,\Eval[\cJ]{ s_n })$.

We typically choose a theory signature $\cStheory$ such that 
$\cStheory \supseteq \cScore=\{\SortBool\}$,
$\Val[{\SortBool[\scriptsize]}] =\{\symb{true}, \symb{false}: \sort{bool}\}$,
$\Sigtheory \supseteq \Sigcore=
\Val[{\SortBool[\scriptsize]}] \cup \{{\wedge}, {\vee}, {\Rightarrow}, {\Leftrightarrow} : \SortBool \times \SortBool \to \SortBool, ~\neg: \SortBool \to \SortBool \} \cup \{ {=_\iota}, {\ne_\iota} : \iota \times \iota \to \SortBool \mid 
\iota \in \cStheory
\}$, 
$\cI(\SortBool) 
= \{\top,\bot\}$,
		and
$\cJ$ interprets these symbols as expected:
	$\cJ(\symb{true}) = \top$ and $\cJ(\symb{false}) = \bot$.
We omit the sort subscripts $\iota$ from $=_\iota$ and $\neq_\iota$ if they are clear from the context.
A theory term with sort $bool$ is called a \emph{constraint}.
A substitution $\gamma$ is said to \emph{respect} a constraint $\phi$ if $x\gamma \in \Val$ for all $x \in \Var(\phi)$ and $\Eval[\cJ]{\phi\gamma} = \top$, where $\Var(\phi)$ denotes the set of variables appearing in $\phi$.
A constraint $\phi$ is said to be \emph{satisfiable} if $\Eval{\phi\gamma} = \top$ for some substitution $\gamma$ respecting $\phi$.

A \emph{constrained rewrite rule} is a triple $ \ell \to r ~[\varphi]$ such that $\ell$ and $r$ are terms of the same sort, $\varphi$ is a constraint, and 
$\ell$ is neither a theory term nor a variable.
If $\varphi = \symb{true}$, then we may write $\ell \to r$.
We define $\LVar(\ell \to r ~ [\varphi])$ as $\Var(\varphi) \cup (\Var(r) \setminus \Var(\ell))$,
the set of \emph{logical variables in $\ell \to r ~ [\varphi]$} which are variables instantiated with value-constants in rewriting terms.
We say that a substitution $\gamma$ \emph{respects} $\ell \to r ~ [\varphi]$ if 
$\gamma(x) \in \Val$ for all $x \in \LVar(\ell \to r ~ [\varphi])$
and $\Eval[\cJ]{\varphi\gamma} = \top$.
Regarding the signature of $\cR$, 
we denote the set $\{ f(x_1,\ldots,x_n) \to y ~ [y = f(x_1,\ldots,x_n)] \mid f 
\in {{\Sigtheory} \setminus {\Val}}, ~
\mbox{$x_1,\ldots,x_n,y$ are pairwise distinct variables}
\}$ by $\cRcalc$.
The elements of $\cRcalc$ are called \emph{calculation rules}
 and we often deal with them as constrained rewrite rules even though their left-hand sides are theory terms.
The \emph{rewrite relation} $\to_{\cR}$ is a binary relation over terms, defined as follows:
For a term $s$, $s[\ell\gamma]_p \mathrel{\to_\cR} s[r\gamma]_p$ if and only if 
$\ell \to r ~ [\varphi] \in \cR \cup \cRcalc$ 
and $\gamma$ respects $\ell \to r ~ [\varphi]$.
A reduction step with $\cRcalc$ is called a \emph{calculation}.
A \emph{logically constrained term rewrite system} (LCTRS, for short) is defined as an abstract reduction system $(T(\Sigma,\cV),{\to_\cR})$, simply denoted by $\cR$. 
An LCTRS is usually given by supplying $\Sigma$, $\cR$, and an informal description of $\cI$ and $\cJ$ if these are not clear from the context.
 The set of \emph{normal forms} of $\cR$ is denoted by $\NF_\cR$.
An LCTRS $\cR$ is said to be \emph{left-linear} if for every rule in $\cR$, the left-hand side is linear.
Note that $\cRcalc$ is left-linear.

The \emph{integer signature} $\Sigint$ is $\Sigcore \cup \{ +, -,\times,\symb{exp},\symb{div},\symb{mod} : \sort{int} \times \sort{int} \to \sort{int} \} \cup \{ {\geq}, {>} : \sort{int} \times \sort{int} \to \SortBool \} \cup {\Val[{\sort[\scriptsize]{int}}]}$ where 
$\cStheory \supseteq \{\sort{int},\SortBool\}$,
$\Val[{\sort[\scriptsize]{int}}] = \{ \symb{n} \mid n \in \mathbb{Z} \}$, $\cI(\sort{int}) = \mathbb{Z}$, and $\cJ(\symb{n}) = n$ for any $n\in\mathbb{Z}$---we use $\symb{n}$ (in \textsf{sans-serif} font) as the value-constant for $n \in \mathbb{Z}$ (in \textit{math} font).
We define $\cJ$ in a natural way.
An LCTRS over a signature $\Sigma$ with $\Sigtheory = \Sigint$ is called an \emph{integer LCTRS}.

\begin{example}
\label{ex:flist-lctrs}
The term $\symb{f}(\symb{cons}(\symb{1},\symb{cons}(\symb{2},\symb{cons}(\symb{0},\symb{cons}(\symb{3},\symb{cons}(\symb{4},\symb{nil}))))),\symb{5})$ is reduced by the LCTRS $\cRflist$ in Example~\ref{ex:flist} to $\symb{4}$ as follows:
$\symb{f}(\symb{cons}(\symb{1},\symb{cons}(\symb{2},\symb{cons}(\symb{0},\linebreak \symb{cons}(\symb{3},\symb{cons}(\symb{4},\symb{nil}))))),\symb{5})
\to_{\cRflist}
\symb{1} + \symb{f}(\symb{cons}(\symb{0},\symb{cons}(\symb{3},\symb{cons}(\symb{4},\symb{nil}))), \symb{5}-\symb{2})
\to_{\cRflist}
\symb{1} + \symb{f}(\symb{cons}(\symb{0},\symb{cons}(\symb{3},\symb{cons}(\symb{4},\symb{nil}))),\symb{3}) 
\to_{\cRflist}
\symb{1} + \symb{f}(\symb{cons}(\symb{3},\symb{cons}(\symb{4},\symb{nil})),\symb{3}-\symb{1}) 
\to_{\cRflist}^*
\symb{4}$.
\end{example}

A function symbol $f$ is called a \emph{defined symbol} of $\cR$ if there exists a rule $f(\ell_1,\ldots,\ell_n) \to r ~ [\varphi] \in \cR \cup \cRcalc$;
non-defined elements of $\Sigma$ are called \emph{constructors} of $\cR$.
Note that all value-constants are constructors of $\cR$.
We denote the sets of defined symbols and constructors of $\cR$ by $\cD_\cR$ and $\cC_\cR$, respectively:
$\cD_\cR = \{ f \mid f(\ldots) \to r ~ [\varphi] \in \cR \cup \cRcalc \}$
and
$\cC_\cR = \Sigma \setminus \cD_\cR$.
A $\cC_\cR$-term 
is called a \emph{constructor term} of $\cR$.
A term of the form $f(t_1,\ldots,t_n)$ with $f\in\cD_\cR$ and $t_1,\ldots,t_n \in T(\cC_\cR,\cV)$ is called a \emph{pattern} (or \emph{basic}) (of $\cR$). 
We call $\cR$ a \emph{constructor system} if the left-hand side of each rule $\ell \to r ~ [\varphi] \in \cR$ is 
a pattern. 
An LCTRS $\cR$ is called \emph{quasi-reducible} if every ground pattern is a redex of $\cR$.

\begin{example}
The LCTRSs $\cRflist$ and $\cRflist'$ in Example~\ref{ex:flist} are left-linear constructor integer LCTRSs,
$\cRflist$ is not quasi-reducible, and $\cRflist'$ is quasi-reducible.
\end{example}

In the rest of the paper, for a signature $\Sigma$, we are interested in defined symbols and constructors without an LCTRS $\cR$ which splits $\Sigma$ to $\cD_\cR$ and $\cC_\cR$.
For this reason, without specifying $\cR$, we often split a signature $\Sigma$ to the sets $\cD,\cC$ of defined symbols and constructors, respectively: $\Sigma = \cD \uplus \cC$.


Finally, we define constrained patterns and their instances.
\begin{definition}[constrained term~\cite{KN13frocos,Kop17,FKN17tocl}]
A \emph{constrained term} over $\Sigma$ is a pair $\CTerm{t}{\psi}$ of a term $t:\iota$ and a constraint $\psi$.
The sort of $\CTerm{t}{\psi}$ is $\iota$.
A constrained term $\CTerm{t}{\psi}$ is called \emph{linear} if $t$ is linear;
$\CTerm{t}{\psi}$ is called 
\emph{linear w.r.t.\ logical variables} (LV-linear, for short) if $t$ is linear w.r.t.\ $\Var(\psi)$;
$\CTerm{t}{\psi}$ is called 
\emph{value-free} if $t \in T(\Sigma\setminus \Val,\cV)$;
$\CTerm{t}{\psi}$ is called 
a \emph{constrained pattern} if $t$ is a pattern.
\end{definition}
\begin{definition}[
ground instances {$\GInst[\cC]{\CTerm{t}{\psi}}$}]
Let 
$\CTerm{t}{\psi}$ be a constrained term.
We denote the set of all ground $\cC$-instances of $\CTerm{t}{\psi}$ by $\GInst[\cC]{\CTerm{t}{\psi}}$:
$
\GInst[\cC]{\CTerm{t}{\psi}} = \{ t\sigma \mid \mbox{$\sigma$ is a $\Var(t)$-ground $\cC$-substitution respecting $\psi$} \}
$.
We extend 
$\GInst[\cC]{\cdot}$ for sets of constrained terms:
$
\GInst[\cC]{P} = \bigcup_{\CTerm{t}{\psi} \in P} \GInst[\cC]{\CTerm{t}{\psi}}
$.
The set of ground $\cC$-instances of an unconstrained term $s$ is denoted by $\GInst[\cC]{s}$:
$\GInst[\cC]{s} = \GInst[\cC]{\CTerm{s}{\symb{true}}}$.
\end{definition}
\begin{convention}[Convention on Term Equivalence]
From the perspective that a (constrained) pattern represents the set of its ground instances, we consider a renamed variant $\CTerm{t'}{\psi'}$ of a constrained term $\CTerm{t}{\psi}$ to be identical with $\CTerm{t}{\psi}$, while we require $\psi'$ to be semantically 
equivalent to $\psi$.
To distinguish such equivalence from the usual one, we write $\CTerm{t}{\psi} \doteq \CTerm{t'}{\psi'}$:
$\CTerm{t}{\psi} \doteq \CTerm{t'}{\psi'}$ if and only if there exists a renaming $\delta$ such that $t = t'\delta$, $\Var(t)\cap\Var(\psi) = \Var(t'\delta) \cap \Var(\psi'\delta)$, and $((\exists \vec{x}.\ \psi) \Leftrightarrow (\exists \vec{y}.\ \psi'\delta))$ is valid, 
where
$\{\vec{x}\} = \Var(\psi) \setminus \Var(t)$
and
$\{\vec{y}\} = \Var(\psi'\delta) \setminus \Var(t'\delta)$.
The equivalence is also used for unconstrained terms:
$t \doteq t'$ if and only if $\CTerm{t}{\symb{true}} \doteq \CTerm{t'}{\symb{true}}$.
Given a set $P$ of constrained terms, 
$\dot{P}$ denote a minimal set that contains exactly one representative for each equivalence class of constrained terms in $P$.
In computing the union of sets of constrained terms, we use the equivalence $\doteq$.
To be more precise, we define the union of (constrained) term sets w.r.t.\ $\doteq$ as follows:
$P \dotcup Q := \dot{U}$, where $U = P \cup Q$.
The disjoint union $\uplus$ w.r.t.\ $\doteq$ is defined as follows:
Given sets $P,Q,U$ of constrained terms, 
$P \dotuplus Q = P \dotcup Q$ if and only if both of the following statements hold:
(i)
for any constrained term $\CTerm{s}{\phi} \in P$, there exists no constrained term $\CTerm{t}{\psi}$ in $Q$ such that $\CTerm{s}{\phi} \doteq \CTerm{t}{\psi}$,
        and
(ii)
for any constrained term $\CTerm{t}{\psi} \in Q$, there exists no constrained term $\CTerm{s}{\phi}$ in $P$ such that $\CTerm{s}{\phi} \doteq \CTerm{t}{\psi}$.
We analogously define $\dotcup,\dotuplus$ for unconstrained terms.
We abuse $\cup$ and $\uplus$ for $\dotcup$ and $\dotuplus$, respectively.
\end{convention}

\section{Complements of Constructor Terms and Substitutions}
\label{sec:construction-of-complements}

In this section, we define complements of constructor terms and constructor substitutions and then recall the operations in~\cite{LM86,LLT90,HA19pro} to compute finite complements of linear constructor terms and constructor substitutions.
In the rest of the paper, we consider a signature $\Sigma = \cD \uplus \cC$ 
without notice.
For examples in this section, we use a subsignature $\Sigmaflist' = \cDflist \cup \cCzlist'$ of $\Sigmaflist$, where $\cDflist=\{\symb{f}\}$ and $\cCzlist' = \{\symb{nil},\symb{cons},\symb{0},\symb{1}\}$.

We first define complements of $\cC$-terms, $\cC$-substitutions, and patterns.
\begin{definition}[complement]
A set $P$ of $\cC$-terms is called a \emph{complement} of a $\cC$-term $u:\iota$ if
$\GInst[\cC]{P} = \{ t \in T(\cC) \mid t:\iota \} \setminus \GInst[\cC]{u}$.
A set $\Theta$ of $\cC$-substitutions is called a \emph{complement} of a $\cC$-substitution $\sigma$ w.r.t.\ a term $t \in T(\Sigma,\cV)$ if 
$\GInst[\cC]{\{ t\rho \mid \rho \in \Theta, \, t\sigma \ne t\rho \}} = \GInst[\cC]{t} \setminus \GInst[\cC]{t\sigma}$.
A set $P$ of patterns is called a \emph{complement} of a pattern $f(s_1,\ldots,s_n)$ if
$\GInst[\cC]{P} = 
\GInst[\cC]{f(x_1,\ldots,x_n)}
\setminus \GInst[\cC]{f(s_1,\ldots,s_n)}$,
where $x_1,\ldots,x_n$ are pairwise distinct variables.
A set $P$ of patterns is called a \emph{complement} of a set $Q$ of patterns if 
$\GInst[\cC]{P} = 
\GInst[\cC]{\{ f(x_1,\ldots,x_n) \mid f \in \cD \}}
\setminus 
\GInst[\cC]{Q}$,
where $x_1,\ldots,x_n$ are pairwise distinct variables.
\end{definition}
Note that complements 
are not unique in general.
Note also that complements of variables and the identity substitution are the empty set.

\begin{example}
The set $\{\symb{nil}, \, \symb{cons}(x,\symb{cons}(x',\var{xs}))\}$ is a complement of $\symb{cons}(y,\symb{nil})$.
The set $\{ \{ \var{xs} \mapsto \symb{nil}, \, y \mapsto \symb{1} \}, \, \{ \var{xs} \mapsto \symb{cons}(x',\var{xs}'), \, y \mapsto \symb{0} \}, \, \{ \var{xs} \mapsto \symb{cons}(x',\var{xs}'), $ $ y \mapsto \symb{1} \} \}$ is a complement of $\{ \var{xs} \mapsto \symb{nil}, \, y \mapsto \symb{0} \}$ w.r.t.\ $\symb{f}(\var{xs},y)$.
The set 
$\{ \symb{f}(\symb{nil},\symb{1}), $ $\symb{f}(\symb{cons}(x',\var{xs}'),y) \}$
is a complement of $\symb{f}(\symb{nil},\symb{0})$.
\end{example}

Next, we recall the constructions of complements of linear $\cC$-terms and $\cC$-substitutions.
We follow the formulation in~\cite{HA19pro}.

For a term $u$, if $\GInst[\cC]{u}$ is infinite, then the linearity of $u$ w.r.t.\ variables with sorts for inductively defined terms such as $\sort{list}$ in Example~\ref{ex:flist} is necessary for finiteness of a complement of $u$ (cf.~\cite[Proposition~4.5]{LM86}).
For this reason, in computing complements of $\cC$-terms, we only consider linear terms.

\begin{definition}[{$\Cocterm[\cC]{u}$} of linear $\cC$-term $u$~\cite{HA19pro}]
\label{def:Cocterm}
For a linear $\cC$-term $u$, we define a set $\Cocterm[\cC]{u}$ of terms recursively as follows:
$\Cocterm[\cC]{x} = \emptyset$ for a variable $x \in \cV$,
        and
$\Cocterm[\cC]{c(u_1,\ldots,u_n)} =
    \{ c'(x_1,\ldots,x_m) \mid c' \ne c, ~ c' : \iota'_1 \times \cdots \times \iota'_m \to \iota \in \cC \}
    \cup \bigcup_{i=1}^n \{ c(u_1,\ldots,u_{i-1},u'_i,y_{i+1},\ldots,y_n) \mid u'_i \in \Cocterm[\cC]{u_i} \}$, 
    where $c : \iota_1 \times \cdots \times \iota_n \to \iota \in \cC$,
    $x_1,\ldots,x_m$ are pairwise distinct variables, and $y_{i+1},\ldots,y_n$ are pairwise distinct fresh variables for any $i \in \{1,\ldots,n\}$.
\end{definition}
For a linear term $u:\iota$, $\Cocterm[\cC]{u}$ is a finite complement of $u$
(i.e., 
$\GInst[\cC]{\Cocterm[\cC]{u}} = \{ t \in T(\cC) \mid t:\iota\} \setminus \GInst[\cC]{u}$),
which is a set of linear $\cC$-terms.

\begin{example}
For $\cCzlist' = \{ \symb{nil}, \symb{cons}, \symb{0}, \symb{1}\}$,
we have that
$\Cocterm[\cCzlist']{\symb{0}} = \{\, \symb{1} \,\}$,
$\Cocterm[\cCzlist']{\symb{1}} = \{\, \symb{0} \,\}$,
$\Cocterm[\cCzlist']{\symb{nil}} = \{\, \symb{cons}(x,\var{xs}) \,\}$,
        and
$\Cocterm[\cCzlist']{\symb{cons}(\symb{0},\symb{cons}(z_3,\var{zs}_3))} = 
    \{\,
    \symb{nil}, \,
    \symb{cons}(\symb{1},ys), \,
    \symb{cons}(\symb{0},\symb{nil})
    \,\}$.
\end{example}


The construction below for $\cC$-substitutions is based on the construction $\Cocterm{\Cdot}$ for linear $\cC$-terms, and thus it works on \emph{linear} substitutions.
A substitution $\sigma$ is called \emph{linear w.r.t.\ $X \subseteq \cV$} ($X$-linear, for short) if 
$x\sigma$ is linear for any variable $x \in X$.
A $\cV$-linear substitution is simply called \emph{linear}.
\begin{definition}[{$\Cosubst[\cC]{\sigma}$} of linear $\cC$-substitution $\sigma$~\cite{HA19pro}]
\label{def:Cosubst}
    Let 
$\sigma$ be a linear $\cC$-substitution.
We define a set $\Cosubst[\cC]{\sigma}$ of substitutions as follows:
$
\Cosubst[\cC]{\sigma} = \{ \rho \mid \Dom(\rho) = \Dom(\sigma), ~ \rho \ne \sigma, 
~ \forall x \in \Dom(\sigma).\ x\rho \in  \Cocterm[\cC]{x\sigma} \cup \{x\sigma\}
\}
$.
\end{definition}
For a linear $\cC$-substitution $\sigma$ and a linear $\cC$-term $t$,
$\{ \rho \mid \rho \in \Cosubst[\cC]{\sigma}, \, t\rho \ne t\sigma \}$ is a finite complement of $\sigma$ w.r.t.\ $t$
(i.e.,
$
\GInst[\cC]{\{ t\rho \mid \rho \in \Cosubst[\cC]{\sigma},\, t\rho \ne t\sigma \}}
=
\GInst[\cC]{t} \setminus \GInst[\cC]{t\sigma}$),
which is a finite set of linear $\cC$-substitutions.

\begin{example}
For $\cCzlist'=\{\symb{nil},\symb{cons},\symb{0},\symb{1}\}$,
we have that 
$\Cosubst[\cCzlist']{\{ \var{xs} \,{\mapsto}\, \symb{nil}, \, y \,{\mapsto}\, \symb{0} \}}
=
\{ \{ \var{xs} \,{\mapsto}\linebreak \symb{nil}, \, y \,{\mapsto}\, \symb{1} \}, \, \{ \var{xs} \,{\mapsto}\, \symb{cons}(x',\var{xs}'), \, y \,{\mapsto}\, \symb{0} \}, \, \{ \var{xs} \,{\mapsto}\, \symb{cons}(x',\var{xs}'), \, y \,{\mapsto}\, \symb{1} \} \}$. 
\end{example}




\section{Difference of Constrained Patterns}
\label{sec:difference}

In this section, we extend the difference operator $\ominus$ over unconstrained linear patterns~\cite{LM86,LLT90,HA19pro} to constrained ones.
Unlike~\cite{LM86,LLT90,HA19pro}, we do not require dividend patterns to be linear, while we still require divisor ones to be linear.

\subsection{Difference Operator for Unconstrained Linear Patterns}

We briefly recall the difference operator $\ominus$ over unconstrained linear patterns.

For unifiable patterns $s,t$ with their mgu $\sigma$ and without any common variables, the idea of $\ominus$ is based on the following property:
\[\textstyle
    \GInst[\cC]{s} = \GInst[\cC]{s\sigma} \uplus 
    \GInst[\cC]{\{ s\rho \mid \rho \in \Cosubst[\cC]{\sigma}, \, s\rho \ne s\sigma \}}
\]
Since $s\sigma = t\sigma$ and 
$\GInst[\cC]{s} \cap \GInst[\cC]{t}=
\GInst[\cC]{t\sigma}$,
the above property implies that 
$ 
\GInst[\cC]{s} \setminus \GInst[\cC]{t} 
=
\GInst[\cC]{\{ s\rho \mid \rho \in \Cosubst[\cC]{\sigma}, \, s\rho \ne s\sigma \}}
$ 
and thus
$\{ s\rho \mid \rho \in \Cosubst[\cC]{\sigma}, \, s\rho \ne s\sigma \}$ is a complement of $s$.
To compute $\Cosubst[\cC]{\sigma}$, the mgu $\sigma$ has to be a $\cC$-substitution.
By definition, it is clear that any mgu of patterns are $\cC$-substitutions.
%
We define the difference operator $\ominus$ for linear patterns as follows.
 
\begin{definition}[$\ominus$ over linear patterns~\cite{HA19pro}]
\label{def:difference-of-unconstrained-terms}
We define the difference operator $\ominus$ over linear patterns as follows:
$s \ominus t =
\{ s\rho \mid \rho \in \Cosubst[\cC]{\sigma}, \, s\rho \ne s\sigma \}$
if $s,t'$ are unifiable,
where 
$t \doteq t'$,
$\Var(s)\cap\Var(t')=\emptyset$, 
and
$\sigma=\mgu(s,t')$;
otherwise, $s \ominus t = \{s\}$.
Note that $s \ominus s = \emptyset$.
\end{definition}

To compute $s \ominus t$, 
$s$ does not have to be linear, while $\sigma|_{\Var(s)}$ must be linear.
A sufficient condition for $\sigma|_{\Var(s)}$ being linear is linearity of $t'$ (i.e., $t$).
\begin{restatable}
{proposition}{PropSufficientConditionOfLinearityOfMgu}
\label{prop:sufficient-condition-of-linearity-of-mgu}
Let $s,t$ be unifiable terms with $\Var(s) \cap \Var(t) = \emptyset$, and $\sigma$ be an mgu of $s,t$. 
If $t$ is linear, then 
$\sigma|_{\Var(s)}$ is linear.
\end{restatable}
\begin{proof}[Sketch]
We generalize this claim to unification problems of the form $\{ u_1 =^? u'_1, \ldots, u_k =^? u'_k \}$ and use the unification procedure in~\cite[Section~4.6]{BN98}.
To distinguish variables in $s$ and $t$, we do not swap both sides of $=^?$ and prepare the \textsf{Eliminate} rule for both sides of $=^?$.
Then, we show that the initial unification problem $\{ s =^? t \}$ can be transformed into an extended solved form $\{ x_1 =^? t_1, \ldots, x_n =^? t_n, s_1 =^? y_1, \ldots, s_m =^? y_m \}$ such that
$x_1,\ldots,x_n$ are pairwise distinct variables in $s$, 
$y_1,\ldots,y_m$ are pairwise distinct variables in $t$,
none of $s_1,\ldots,s_m$ is a variable,
$\{x_1,\ldots,x_n,y_1,\ldots,y_m\} \cap \Var(t_1,\ldots,t_n,s_1,\ldots,s_m) = \emptyset$,
and 
$t_i$ is linear for all $1 \leq i \leq n$.
A complete proof can be seen in the appendix.
\qed
\end{proof}
For a pattern $s$ and a linear pattern $t$, $s \ominus t$ is a finite complement of patterns with sort $\iota$ w.r.t.\ ground $\cC$-instances of $s$
(i.e., $\GInst[\cC]{s \ominus t} = \GInst[\cC]{s} \setminus \GInst[\cC]{t}$).

\begin{example}
For $\Sigmaflist'=\{\symb{f},\symb{nil},\symb{cons},\symb{0},\symb{1}\}$,
we have that
$ \symb{f}(\var{xs},y) \ominus \symb{f}(\symb{nil},\symb{0})
= 
\{ \symb{f}(\symb{nil},\symb{1}), \, \symb{f}(\symb{cons}(x',\var{xs}'),\symb{0}), \, \symb{f}(\symb{cons}(x',\var{xs}'),\symb{1}) \}
$.
\end{example}

\subsection{Extension of Difference Operator to Constrained Patterns}
\label{subsec:extension-of-ominus-to-constrained-patterns}



In the definition of $\ominus$ for unconstrained patterns, for $s \ominus t$, unifiability of $s$ and $t$ is used.
Thus, for the extension, 
we define unifiability of constrained terms.

\begin{definition}[unifiability of constrained terms]
\label{def:unifiability-of-constrained-terms}
Two constrained terms $\CTerm{s}{\phi}$ and $\CTerm{t}{\psi}$ are said to be \emph{unifiable} if $s$ and $t$ are unifiable and their mgu $\theta$ with $\Dom(\theta) \subseteq \Var(s,t)$ satisfies both of the following:
$x\theta \in \Val \cup \cV$ for all variables $x \in \Var(\phi,\psi)$,
        and
$\phi\theta \land \psi\theta$ is satisfiable.
\end{definition}
\begin{example}
The constrained terms $\CTerm{\symb{f}(\symb{nil},y_1)}{y_1 \leq \symb{0}}$ and $\CTerm{\symb{f}(\symb{nil},y_{\mathrm{a}})}{y_{\mathrm{a}} > \symb{0}}$ are not unifiable, but $\CTerm{\symb{f}(\symb{nil},y_1)}{y_1 \leq \symb{0}}$ and $\CTerm{\symb{f}(\var{xs},y)}{\symb{true}}$ are unifiable by means of an mgu $\{ y_1 \mapsto y, \var{xs} \mapsto \symb{nil} \}$ of $\symb{f}(\symb{nil},y_1)$ and $\symb{f}(\var{xs},y)$;
their instance by means of the mgu is $\CTerm{\symb{f}(\symb{nil},y)}{y \leq \symb{0} \land \symb{true}}$.
\end{example}

In considering ground constructor instances of a constrained term $\CTerm{s}{\phi}$, logical variables---variables in $\phi$---are instantiated by value-constants but others are instantiated by arbitrary ground terms.
Under our assumption that there is no user-defined constructor term with a theory sort,
any variable with a theory sort is instantiated by value-constants only,
and no variable with a non-theory sort is instantiated by any value-constant.
For this reason, unlike the usual discussion on the instantiation of variables in the LCTRS setting (e.g.,~\cite{ANS24fscd}), we do not take care of whether a variable is logical or not.

However, signatures of LCTRSs are infinite in general.
For example, signatures of integer LCTRSs include the integers and are thus infinite.
For this reason, complements are often infinite, e.g., a complement of the value-constant $\symb{0}$ is $\{ \symb{n} \mid n \in \mathbb{Z} \setminus \{0\} \}$.
On the other hand, user-defined constructors (i.e., $\cC \setminus \Val$) are usually finite.
Since we consider constrained patterns, we can move value-constants in terms to constraints.
For example, we transform the constrained pattern $\CTerm{\symb{f}(\symb{nil},\symb{0})}{\symb{true}}$ into $\CTerm{\symb{f}(\symb{nil},x)}{x = \symb{0}}$, where $x$ is a fresh variable.
In summary, we consider value-free constrained patterns. 
Since each constrained term has an equivalent value-free constrained term~\cite{Kop17,KN24jip}, in the rest of this paper, we assume w.l.o.g.\ that constrained patterns are value-free.
For a constrained term $\CTerm{s}{\phi}$, we denote its value-free variant by $\CTerm{\tilde{s}}{\tilde{\phi}}$~\cite{KN24jip}:
$\tilde{s} = s[y_1,\ldots,y_n]_{p_1,\ldots,p_n}$ and
$\tilde{\phi} = \phi \land \bigwedge_{i=1}^n (y_i = s|_{p_i})$,
where $p_1,\ldots,p_n$ are the positions of values in $s$, and $y_1,\ldots,y_n$ are pairwise distinct variables such that $\{y_1,\ldots,y_n\} \cap \Var(s,\phi) = \emptyset$.
It is clear that $\GInst[\cC]{\CTerm{s}{\phi}} = \GInst[\cC]{\CTerm{\tilde{s}}{\tilde{\phi}}}$.

For the extension of $\ominus$ to constrained patterns, the key property of $\ominus$---%
$
    \GInst[\cC]{s} 
    = \GInst[\cC]{s\sigma} \uplus 
     \GInst[\cC]{\{ s\rho \mid \rho \in \Cosubst[\cC]{\sigma}, \, s\rho \ne s\sigma \}}
$---%
is adapted to unifiable constrained patterns $\CTerm{s}{\phi}, \CTerm{t}{\psi}$ with $\sigma=\mgu[\cC](s,t)$ as follows:
\[\textstyle
    \GInst[\cC]{\CTerm{s}{\phi}} = 
    \GInst[\cC]{\CTerm{s\sigma}{\phi\sigma \land \psi\sigma}} 
    \uplus
    \GInst[\cC]{\CTerm{s\sigma}{\phi\sigma \land \neg\psi\sigma}} 
    \uplus
    \GInst[\cC]{\{ \CTerm{s\rho}{\phi\sigma} \mid \rho \in \Cosubst[\cC]{\sigma}, \, s\rho \ne s\sigma\}}
\]
    The first and second ones are obtained by the following division of $\GInst[\cC]{\CTerm{s\sigma}{\phi\sigma}}$:
$ 
    \GInst[\cC]{\CTerm{s\sigma}{\phi\sigma}} = 
    \GInst[\cC]{\CTerm{s\sigma}{\phi\sigma \land \psi\sigma}} \uplus 
    \GInst[\cC]{\CTerm{s\sigma}{\phi\sigma \land \neg\psi\sigma}}
$. 
    The third one can be obtained similarly to the unconstrained case, but
    we restrict instances to those satisfying $\phi\sigma$.
    This is because all ground $\cC$-instances in $\CTerm{s}{\phi}$ satisfy $\phi$, i.e., a substitution to obtain a ground $\cC$-instance in $\CTerm{s}{\phi}$ respects $\phi$.
    Since $\CTerm{s}{\phi}$ and $\CTerm{t}{\psi}$ are assumed to be value-free, 
    we have that 
    $\Ran(\sigma|_{\Var(\phi,\psi)}) \subseteq \cV$.
    Thus, by definition, for all substitutions $\rho \in \Cosubst[\cC]{\sigma}$, we have that $\rho|_{\Var(\phi,\psi)} = \sigma|_{\Var(\phi,\psi)}$, and hence $\phi\sigma = \phi\rho$.

Following the idea above, we extend the difference operator $\ominus$ to value-free constrained patterns.

\begin{definition}[$\ominus$ over value-free constrained patterns]
We define a difference operator $\ominus$, which takes a value-free constrained pattern as a dividend and takes a value-free constrained linear pattern as a divisor, as follows:
\[
{\CTerm{s}{\phi}} \,\ominus\, {\CTerm{t}{\psi}} =
\left\{
\begin{array}{l@{\hspace{-12ex}}l@{}}
\{ \CTerm{s\rho}{\phi\sigma} \mid \rho \in \Cosubst[\cC]{\sigma}, \, s\rho \ne s\sigma\}
\\
{} \cup\{\CTerm{s\sigma}{\phi\sigma\land\lnot\psi'\sigma} \mid \mbox{$(\phi\sigma\land\lnot\psi'\sigma)$ is satisfiable} \}
\\
& \mbox{if $\CTerm{s}{\phi},\CTerm{t'}{\psi'}$ are unifiable} 
\\[3pt]
\{\CTerm{s}{\phi}\} 
& \mbox{otherwise} 
\\
\end{array}
\right.
\]
where 
$\CTerm{t'}{\psi'} \doteq \CTerm{t}{\psi}$, $\Var(s,\phi)\cap\Var(t',\psi')=\emptyset$,
and $\sigma = \mgu(s,t')$.
\end{definition}
By definition, it is clear that $\CTerm{s}{\phi} \ominus \CTerm{s'}{\phi'} = \emptyset$
for any value-free constrained pattern $\CTerm{s'}{\phi'}$ such that $\CTerm{s'}{\phi'} \doteq \CTerm{s}{\phi}$.

\begin{example}
For $\Sigmaflist=\{\symb{f},\symb{nil},\symb{cons}\}\cup \{ \symb{n} \mid n \in \mathbb{Z} \}$, we have that
${\CTerm{\symb{f}(\var{xs},y)}{\symb{true}}} \ominus {\CTerm{\symb{f}(\symb{nil},y_1)}{y_1 \leq \symb{0}}}
= 
\{\, \CTerm{\symb{f}(\symb{nil},y_{\mathrm{a}})}{\neg (y_{\mathrm{a}} \leq \symb{0})}, \, \CTerm{\symb{f}(\symb{cons}(x',\var{xs}'),y')}{\symb{true}} \,\}
$.
\end{example}

The following theorem shows correctness of $\ominus$ over constrained patterns.

\begin{restatable}[correctness of $\CTerm{s}{\phi} \ominus \CTerm{t}{\psi}$]
{theorem}{ThmConstrainedPatternOminusPropertiesX}
\label{thm:constrained-pattern-ominus-propertiesX}
For a value-free constrained pattern $\CTerm{s}{\phi}:\iota$ and a value-free constrained linear pattern $\CTerm{t}{\psi}$,
$\CTerm{s}{\phi} \ominus \CTerm{t}{\psi}$ is a finite set of value-free constrained patterns with sort $\iota$
such that
\begin{enumerate}
    \item $\GInst[\cC]{\CTerm{u}{\varphi}} \cap \GInst[\cC]{\CTerm{u'}{\varphi'}} = \emptyset$ for any different constrained patterns $\CTerm{u}{\varphi},\CTerm{u'}{\varphi'} \in 
    (\CTerm{s}{\phi} \ominus \CTerm{t}{\psi})$,
    \item if $\CTerm{s}{\phi}$ is linear, then $\CTerm{s}{\phi} \ominus \CTerm{t}{\psi}$ is a set of constrained linear patterns,
    \item $s < u$ for any constrained term $\CTerm{u}{\varphi} \in (\CTerm{s}{\phi} \ominus \CTerm{t}{\psi})$,
    \item $\max\{\Height{s},\Height{t}\} \geq \Height{u}$ for any term $\CTerm{u}{\varphi} \in (\CTerm{s}{\phi} \ominus \CTerm{t}{\psi})$,
        and
    \item $\GInst[\cC]{\CTerm{s}{\phi} \ominus \CTerm{t}{\psi}} = \GInst[\cC]{\CTerm{s}{\phi}} \setminus \GInst[\cC]{\CTerm{t}{\psi}}$.
\end{enumerate}
\end{restatable}
\begin{proof}[Sketch]
This can be proved as a straightforward extension of the correctness proof for unconstrained patterns to constrained ones by means of the following properties:
    $\{ s\rho \mid \rho \in \Cosubst[\cC]{\sigma}, \, s\rho \ne s\sigma \}$ 
    is a set of value-free patterns;
$\sigma|_{\Var(\phi,\psi')} = \rho|_{\Var(\phi,\psi')}$ for any $\rho \in \Cosubst[\cC]{\sigma}$;
$
    \GInst[\cC]{\CTerm{s\sigma}{\phi\sigma}} = 
    \GInst[\cC]{\CTerm{s\sigma}{\phi\sigma \land \psi'\sigma}} \uplus 
    \GInst[\cC]{\CTerm{s\sigma}{\phi\sigma \land \neg\psi'\sigma}}
$ for an arbitrary constraint $\psi'$.
A complete proof can be seen in the appendix.
\qed
\end{proof}

\section{Complements of Constrained Patterns}
\label{sec:complements-of-constrained-patterns}

In this section, we first extend the complement algorithm for unconstrained linear patterns in~\cite{LM86,LLT90,HA19pro} to constrained patterns.
Then, using the extended complement algorithm, we show that quasi-reducibility is decidable for left-linear LCTRSs over a signature $\Sigma$ with decidable built-in theories 
such that $\Sigterm$ is finite and
there is no constructor $c: \iota_1 \times \cdots \times \to \iota \in \cC_\cR$ with $\iota \in \cStheory$.


Complements of constrained patterns are defined as a straightforward extension of complements of unconstrained patterns.
\begin{definition}[complement of constrained patterns]
A set $P$ of constrained patterns is called a \emph{complement} of a set $P'$ of constrained patterns \emph{w.r.t.} an $n$-ary defined symbol $f \in \cD$ 
if
$\GInst[\cC]{P} = 
\GInst[\cC]{\{ f(x_1,\ldots,x_n) \}}
\setminus \GInst[\cC]{P'}$,
where $x_1,\ldots,x_n$ are pairwise distinct variables.
A set $P$ of constrained patterns is called a \emph{complement} of a constrained pattern $\CTerm{f(s_1,\ldots,s_n)}{\phi}$ 
if
$P$ is a complement of $\{\CTerm{f(s_1,\ldots,s_n)}{\phi}\}$ w.r.t.\ $f$.
A set $P$ of constrained patterns is called a \emph{complement} of a set $P'$ of constrained patterns if 
$\GInst[\cC]{P} = 
\GInst[\cC]{\{ f(x_1,\ldots,x_n) \mid f \in \cD \}}
\setminus 
\GInst[\cC]{P'}$,
where $x_1,\ldots,x_n$ are pairwise distinct variables.
\end{definition}
Note that as for unconstrained patterns, complements of constrained patterns are not unique.

\begin{example}
Let us consider the $\cSzlist$-sorted signature $\Sigmaflist$ in Example~\ref{ex:flist} again.
The set $\{ (\mathrm{a})~\CTerm{\symb{f}(\symb{nil},y_{\mathrm{a}})}{\neg(y_{\mathrm{a}} \leq \symb{0}}), \, (\mathrm{b})~\CTerm{\symb{f}(\symb{cons}(x_{\mathrm{b}},\symb{nil}),y_{\mathrm{b}})}{\symb{true}} \}$ is a complement of~$(1)~\CTerm{\symb{f}(\symb{nil},y_1)}{y \leq \symb{0}}$.
\end{example}

\subsection{Difference Operator over Sets of Constrained Patterns}

The complement algorithm in~\cite{LM86,LLT90,HA19pro} is defined by the repetition of applying the difference operator $\ominus$ over patterns.
Thus, as in~\cite{LM86,LLT90,HA19pro}, we extend $\ominus$ to finite sets of constrained patterns
in order to have 
the results in Example~\ref{ex:goal}.

In computing $P \ominus Q$, 
if there exist constrained patterns $\CTerm{s}{\phi} \in P$ and $\CTerm{t}{\psi} \in Q$ such that $\CTerm{s}{\phi} \ominus \CTerm{t}{\psi} \ne \{\CTerm{s}{\phi}\}$, then
we recursively compute
$((P \setminus \{\CTerm{s}{\phi}\}) \cup (\CTerm{s}{\phi} \ominus \CTerm{t}{\psi})) \ominus ((Q \setminus \{\CTerm{t}{\psi}\}) \cup (\CTerm{t}{\psi} \ominus \CTerm{s}{\phi}))$.
In Section~\ref{subsec:extension-of-ominus-to-constrained-patterns}, we required divisors of $\ominus$ to be linear, and thus $\CTerm{t}{\psi}$ has to be linear.
It is difficult to know which constrained terms are selected as a divisor for $\ominus$ over constrained patterns.
For this reason, to ensure the linear of divisors for the application of $\ominus$ to constrained patterns, we assume that all constrained terms in $Q$ are linear, while not all of them may be necessarily linear.
Then, for the recursive call of $\ominus$ for sets of constrained patterns,
we would like to ensure that all constrained patterns in 
$(Q \setminus \{\CTerm{t}{\psi}\}) \cup (\CTerm{t}{\psi} \ominus \CTerm{s}{\phi})$ are linear.
Since $Q$ is assumed to be linear in advance, we need to ensure that all constrained patterns in $\CTerm{t}{\psi} \ominus \CTerm{s}{\phi}$ are linear.
Therefore, for the extension of $\ominus$ to sets of constrained patterns, we assume that all constrained patterns in $P$ are linear, as well as $Q$.

\begin{definition}[$\ominus$ over sets of value-free constrained linear patterns]
We extend $\ominus$ for finite sets of value-free constrained linear patterns as follows:
For finite sets $P,Q$ of value-free constrained linear patterns,
\[
P \mathop{\ominus} Q = 
\!\left\{
\begin{array}{@{}l@{\hspace{-36.5ex}}l@{}}
(P' \cup (\CTerm{s}{\phi} \ominus \CTerm{t}{\psi})) \ominus (Q' \cup (\CTerm{t}{\psi} \ominus \CTerm{s}{\phi})) 
\\
& \mbox{if $P = P' \uplus \{\CTerm{s}{\phi}\}$, $Q = Q' \uplus \{\CTerm{t}{\psi}\}$,
and
$\CTerm{s}{\phi} \ominus \CTerm{t}{\psi} \ne \{\CTerm{s}{\phi}\}$%
}
\\[3pt]
P
& \mbox{otherwise}
\end{array}
\right.
\]
    
\end{definition}
Since $\ominus$ over sets of value-free constrained linear patterns is non-deterministic, there may be two or more sets that can be results of $P \ominus Q$, while all the results are correct.
In addition, a result of $P \ominus Q$ is not always minimal. 
Note that 
using $\Cosubst[\cC]{\sigma}$, 
$\CTerm{s}{\phi} \ominus \CTerm{t}{\psi}$ and $\CTerm{t}{\psi} \ominus \CTerm{s}{\phi}$ 
can be computed simultaneously.

\begin{example}
For $\Sigmaflist=\{\symb{f},\symb{nil},\symb{cons}\}\cup \{ \symb{n} \mid n \in \mathbb{Z} \}$, we have that
\[
\begin{array}{@{}l@{}c@{}l@{}}
\lefteqn{\{ \CTerm{\symb{f}(\var{xs},y)}{\symb{true}} \}
\ominus 
\left\{
\begin{array}{@{}l@{~}r@{\,}l@{}}
(1) & \CTerm{\symb{f}(\symb{nil},y_1) &}{y_1 \leq \symb{0}},\\
(2) & \CTerm{\symb{f}(\symb{cons}(x_2,\var{xs}_2), y_2) &}{x_2 \leq \symb{0} \land y_2 > \symb{0}},\\
(3) & \CTerm{\symb{f}(\symb{cons}(x_3,\symb{cons}(z_3,\var{zs}_3)), y_3) &}{ x_3 > \symb{0} \land y_3 > \symb{1} }\\
\end{array}
\right\}
}\\
&=& 
\left\{ 
\begin{array}{@{}l@{~}r@{\,}l@{}}
(\mathrm{a}) & \CTerm{\symb{f}(\symb{nil},y_{\mathrm{a}}) &}{\neg(y_{\mathrm{a}} \leq \symb{0})},  \\
(1') & \CTerm{\symb{f}(\symb{cons}(x',\var{xs}'),y') &}{\symb{true}} \\
\end{array}
\right\}
\ominus
\{\,(2), \, (3) \,\}
\\
&=& 
\left\{ 
\begin{array}{@{}l@{~}r@{\,}l@{}}
(\mathrm{a}) & \CTerm{\symb{f}(\symb{nil},y_{\mathrm{a}}) &}{\neg(y_{\mathrm{a}} \leq \symb{0})},  \\
(2') & \CTerm{\symb{f}(\symb{cons}(x',\var{xs}'),y') &}{\neg(x' \leq \symb{0} \land y' > \symb{0})} \\
\end{array}
\right\}
\ominus
\{\, (3) \,\}
\\
&=& 
\left\{ 
\begin{array}{@{}l@{~}r@{\,}l@{}}
(\mathrm{a}) & \CTerm{\symb{f}(\symb{nil},y_{\mathrm{a}}) &}{\neg(y_{\mathrm{a}} \leq \symb{0})},  \\
(\mathrm{b}) & \CTerm{\symb{f}(\symb{cons}(x_{\mathrm{b}},\symb{nil}),y_{\mathrm{b}}) &}{\symb{true} }, \\
(\mathrm{c}) & \CTerm{\symb{f}(\symb{cons}(x_{\mathrm{c}},\symb{cons}(z_{\mathrm{c}},\var{zs}_{\mathrm{c}})),y_{\mathrm{c}}) &}{ 
\neg(x_{\mathrm{c}} \,{\leq}\, \symb{0} \,{\land}\, y_{\mathrm{c}} \,{>}\, \symb{0})
\,{\land}\,
\neg(x_{\mathrm{c}} \,{>}\, \symb{0} \,{\land}\, y_{\mathrm{c}} \,{>}\, \symb{1})} \\
\end{array}
\right\} \!
\ominus
\emptyset
\\
&=& 
\{ (\mathrm{a}), \, (\mathrm{b}), \, (\mathrm{c}) \}
\\
\end{array}
\]
\end{example}

A key for correctness of $\ominus$ over sets of (constrained) linear patterns is non-overlappingness between (constrained) patterns in dividends sets of $\ominus$.
We define a notion of non-overlappingness of constrained patterns and their sets.
\begin{definition}[$\cC$-non-overlappingness of constrained terms]
Two constrained terms $\CTerm{s}{\phi},\CTerm{t}{\psi}$ are called \emph{non-overlapping w.r.t.\ $\cC$} ($\cC$-non-overlapping, for short) if $\GInst[\cC]{\CTerm{s}{\phi}} \cap \GInst[\cC]{\CTerm{t}{\psi}} = \emptyset$.
Two sets $P,Q$ of constrained terms are said to be \emph{$\cC$-non-overlapping} if $\GInst[\cC]{P} \cap \GInst[\cC]{Q} = \emptyset$.
A set $P$ of constrained terms is said to be \emph{pairwise $\cC$-non-overlapping} if 
any two different constrained terms $\CTerm{s}{\phi},\CTerm{t}{\psi} \in P$ are $\cC$-non-overlapping.
\end{definition}

\begin{example}
The constrained terms $\CTerm{\symb{f}(\symb{nil},y_1)}{y_1 \leq \symb{0}}$ and $\CTerm{\symb{f}(\symb{nil},y_{\mathrm{a}})}{\neg(y_{\mathrm{a}} \leq \symb{0})}$ are $\cCzlist$-non-overlapping,
but $\CTerm{\symb{f}(\symb{nil},y_1)}{y_1 \leq \symb{0}}$ and $\CTerm{\symb{f}(\var{xs},y)}{\symb{true}}$ are not $\cCzlist$-non-overlapping because, e.g.,
$\symb{f}(\symb{nil},\symb{0}) \in \GInst[\cCzlist]{\CTerm{\symb{f}(\symb{nil},y_1)}{y_1 \leq \symb{0}}} \cap
\GInst[\cCzlist]{\CTerm{\symb{f}(\var{xs},y)}{\symb{true}}}$.
\end{example}
Note that non-overlappingness and unifiability of constrained patterns are dual properties:
$\CTerm{s}{\phi},\CTerm{t}{\psi}$ are $\cC$-non-overlapping 
if and only if 
$\CTerm{s}{\phi}$ and $\CTerm{t'}{\psi'}$ are not unifiable,
where $\CTerm{t'}{\psi'} \doteq \CTerm{t}{\psi}$ and
$\Var(s,\phi) \cap \Var(t',\psi') = \emptyset$.


For correctness of $P \ominus Q$, we assume that $P$ is pairwise $\cC$-non-overlappingness.
\begin{example}
For $\Sigmaflist=\{\symb{f},\symb{nil},\symb{cons}\}\cup \{ \symb{n} \mid n \in \mathbb{Z} \}$,
we have that a result of $\{ \CTerm{\symb{f}(\symb{nil},y)}{\symb{true}}, \, \CTerm{\symb{f}(\var{xs},y)}{\symb{true}} \} \ominus \{ \CTerm{\symb{f}(\var{xs}',y')}{\symb{true}} \}$ is $\{ \CTerm{\symb{f}(\symb{nil},y)}{\symb{true}} \}$, which is not an expected one.
Since $\symb{f}(\var{xs},y) \lesssim \symb{f}(\symb{nil},y)$, the set $\{ \CTerm{\symb{f}(\symb{nil},y)}{\symb{true}}, \, \CTerm{\symb{f}(\var{xs},y)}{\symb{true}} \}$ is redundant from the viewpoint of ground instances.
\end{example}

The following proposition shows the correctness of $\ominus$ over sets of value-free constrained linear patterns.

\begin{restatable}[correctness of $P \ominus Q$]
{theorem}{ThmConstrainedSetOminusProperties}
\label{thm:constrained-set-ominus-properties}
Let 
$P,Q$ be finite sets of value-free constrained linear patterns.
If $P$ is pairwise $\cC$-non-overlapping, 
then
$P \ominus Q$ is a pairwise $\cC$-non-overlapping finite set of value-free constrained linear patterns
such that
$\GInst[\cC]{P \ominus Q} = \GInst[\cC]{P} \setminus \GInst[\cC]{Q}$.
\end{restatable}
\begin{proof}[Sketch]
To prove this claim, we show that 
for finite sets $P,Q$ of value-free constrained linear patterns,
all of the following hold:
(i)
    if $P$ is pairwise $\cC$-non-overlapping, $P = P'\uplus \{\CTerm{s}{\phi}\}$, $Q=Q' \uplus \{\CTerm{t}{\psi}\}$, and
    $\CTerm{s}{\phi} \ominus \CTerm{t}{\psi} \ne \{\CTerm{s}{\phi}\}$,
    then 
        $P'$ and $\{\CTerm{s}{\phi} \ominus \CTerm{t}{\psi}\}$ are $\cC$-non-overlapping,
        and
(ii)
    if $P$ is pairwise $\cC$-non-overlapping, then $\GInst[\cC]{P \ominus Q} = \GInst[\cC]{P} \setminus \GInst[\cC]{Q}$.

The first statement~(i) is a straightforward extension of the corresponding statement for unconstrained linear patterns to constrained ones.
The difficulty of the second statement~(ii) is the well-founded order for induction.
Let $h$ be the maximum height of terms in $P,Q$.
Then, in computing $P \ominus Q$, the heights of all terms considered during the computation are less than or equal to $h$, and thus such terms are finitely many.
We denote by $T(\Sigma,\cV)_{\leq h}$ the set of terms whose heights are less than or equal to $h$.
Then, We define a quasi-order $\SuccsimH$ over terms in $T(\Sigma,\cV)_{\leq h}$ as follows:
$s \SuccsimH t$ if and only if $s \lesssim t$, $\Height{s} \leq h$, and $\Height{t} \leq h$.
The strict part $\SuccH$ of $\SuccsimH$ is defined as ${\SuccH} = ({\SuccsimH} \setminus {\PrecsimH})$.
Note that $s \SuccH t$ if and only if $s < t$, $\Height{s} \leq h$, and $\Height{t} \leq h$.
In computing $\CTerm{s}{\phi} \ominus \CTerm{t}{\psi}$, the resulting set may contain $\CTerm{s\sigma}{\phi\sigma \land \neg\psi'\sigma}$ and we may have that $s = s\sigma$ (and thus $(\phi\sigma \land \neg\psi'\sigma) = (\phi \land \neg\psi')$.
On the other hand, since $\phi \land \psi'$ and $\neg\psi'$ are satisfiable, we have that
$\{ \theta \mid \Eval{\phi\theta} = \top \} \supset \{ \theta \mid \Eval{(\phi \land \neg\psi')\theta} = \top \}$.
However, the sets may be infinite, and thus the relation $\supset$ for the sets is not well-founded in general.
To overcome the problem, we consider the number of constrained terms $\CTerm{t}{\psi}$ 
such that $\CTerm{s}{\phi} \ominus \CTerm{t}{\psi} \ne \{ \CTerm{s}{\phi}\}$.
Let $\SuccsimHN$ be the order of lexicographic products of terms and natural numbers compared by $\SuccsimH$ and $\geq_{\mathbb{N}}$.
Let ${\SuccHN} = (\SuccsimHN \setminus \PrecsimHN)$.
Then, it is clear that $\SuccHN$ is well-founded.
We define the weight $w$ for pairs of sets of constrained linear patterns as follows:
$
w(P,Q) = \{ (s,n) \mid \CTerm{s}{\phi} \in P, \, n = |\{ \CTerm{t}{\psi} \in Q \mid 
\CTerm{s}{\phi} \ominus \CTerm{t}{\phi} \ne \{\CTerm{s}{\phi}\}
\}| \,\}
$.
When $P \ominus Q$ calls $P' \ominus Q'$,
we have that $w(P,Q) \SuccHN w(P',Q')$.
The second statement~(ii) can be proved by induction on the well-founded order $\SuccHN$ with the weight $w$.
A complete proof can be seen in the appendix.
\qed
\end{proof}

Note that for $P \ominus Q$, if there are no constrained terms $\CTerm{s}{\phi} \in P$ and $\CTerm{t}{\psi} \in Q$ such that $\CTerm{s}{\phi} \in P$ and $\CTerm{t}{\psi} \in Q$ are not unifiable, 
then neither $P$ nor $Q$ must be a set of constrained linear patterns
and
$P$ must not be pairwise $\cC$-non-overlapping.

\subsection{Complement Algorithm for Constrained Linear Patterns}

In this section, we first show a complement algorithm for sets of value-free constrained linear patterns, and then show that quasi-reducibility is decidable for some class of left-linear LCTRSs with decidable built-in theories.

\begin{definition}
Let 
$f$ be an $n$-ary defined symbol in $\cD$,
and $Q$ be a finite set of value-free  constrained linear patterns. 
Then, we define a finite set $\CopatternSet[f]{Q}$ of constrained patterns as follows:
$
\CopatternSet[f]{Q} = \{ \CTerm{f(x_1,\ldots,x_n)}{\symb{true}} \} \ominus Q
$,
where $x_1,\ldots,x_n$ are pairwise distinct variables.
We extend $\CopatternSet[f]{\cdot}$ for $\cD$ as follows:
$
\CopatternSet{Q} = \bigcup_{f \in \cD} \CopatternSet[f]{Q}
$.
\end{definition}
\begin{example}
For $\Sigmaflist=\{\symb{f},\symb{nil},\symb{cons}\}\cup \{ \symb{n} \mid n \in \mathbb{Z} \}$ in Example~\ref{ex:flist}, we have that 
$\CopatternSet[]{\{(1),(2),(3)\}} = \CopatternSet[\symb{f}]{\{(1),(2),(3)\}} = \{(\mathrm{a}), (\mathrm{b}),(\mathrm{c})\}$.
\end{example}

As a direct consequence of Theorem~\ref{thm:constrained-set-ominus-properties},
the following claim holds for $\CopatternSet[f]{\cdot}$.

\begin{restatable}{theorem}{ThmOptimizationOfCopatternSetX}
\label{thm:optimization-of-CopatternSetX}
Let 
$Q$ be a pairwise $\cC$-non-overlapping finite set of value-free constrained linear patterns, and 
$f$ an $n$-ary defined symbol in $\cD$. 
Then, $
\GInst[\cC]{\CopatternSet[f]{Q}}$ $=
\GInst[\cC]{\CTerm{f(x_1,\ldots,x_n)}{\symb{true}}} \setminus \GInst[\cC]{\{ \CTerm{t}{\psi} \in Q \mid \Root(t) = f \}}$
and $\CopatternSet[f]{Q}$ is a complement of $Q$ w.r.t.\ $f$,
where $x_1,\ldots,x_n$ are pairwise distinct variables.
\end{restatable}
\begin{proof}[Sketch]
By definition, during the computation of $\CopatternSet[f]{Q}$, all constrained patterns in the left-hand sides of $\ominus$ are rooted by $f$.
Thus, no constrained term $\CTerm{t}{\psi}$ in $Q$ such that $\Root(t) \ne f$ is not unifiable with any constrained pattern in the left-hand sides of $\ominus$, and hence
$\{ \CTerm{f(x_1,\ldots,x_n)}{\symb{true}} \} \ominus Q = \{ \CTerm{f(x_1,\ldots,x_n)}{\symb{true}} \} \ominus \{ \CTerm{t}{\psi} \in Q \mid \Root(t) = f \}$.
Then, by Theorem~\ref{thm:constrained-set-ominus-properties}, $\CopatternSet[f]{Q}$ is a complement of $Q$ w.r.t.\ $f$.
A complete proof can be seen in the appendix.
\qed
\end{proof}
The following corollary is a direct consequence of Theorem~\ref{thm:optimization-of-CopatternSetX}.
\begin{corollary}
\label{cor:correctness-of-CopatternSetX}
For a pairwise $\cC$-non-overlapping finite set $Q$ of value-free constrained linear patterns, 
    $
    \GInst[\cC]{\CopatternSet{Q}}
    =
    \GInst[\cC]{\{f(x_1,\ldots,x_n) \mid f \in \cD \}}
    \setminus
    \GInst[\cC]{Q}$,
where $x_1,\ldots,x_n$ are pairwise distinct variables.
\end{corollary}


Corollary~\ref{cor:correctness-of-CopatternSetX} immediately implies decidability of quasi-reducibility for left-linear LCTRSs with decidable built-in theories.
\begin{restatable}{theorem}{ThmDecidabilityOfQuasiReducibilityOfLLLCTRSs}
\label{thm:decidability-of-quasi-reducibility-of-LL-LCTRSs}
Let $\cR$ be a finite left-linear LCTRS 
such that $\Sigterm$ is finite and there is no constructor $c: \iota_1 \times \cdots \times \to \iota \in \cC_\cR$ with $\iota \in \cStheory$.
Then, $\cR$ is quasi-reducible if and only if $\CopatternSet{\{ \CTerm{\tilde{\ell}}{\tilde{\phi}} \mid \ell \to r ~[\phi] \in \cR, ~ \mbox{$\ell$ is a pattern}\}} = \emptyset$.
Thus, quasi-reducibility is decidable for such LCTRSs with decidable built-in theories.%
\footnote{It is sufficient that all constraints of constrained rewrite rules in $\cR$ are in a class of a decidable theory.
For example, the integer theory is not decidable, but the constraints of $\cRflist$ and $\cRflist'$ are formulas of Presburger arithmetic.
}
\end{restatable}
\begin{proof}[Sketch]
Since the built-in theory is decidable, 
by the well-founded order in the proof of Theorem~\ref{thm:constrained-set-ominus-properties},
the set $\CopatternSet{\{ \CTerm{\tilde{\ell}}{\tilde{\phi}} \mid \ell \to r ~[\phi] \in \cR, ~ \mbox{$\ell$ is a pattern}\}}$ is computable, and thus
its emptiness is decidable.
Therefore, quasi-reducibility of $\cR$ is decidable.
A complete proof can be seen in the appendix.
\qed
\end{proof}

\begin{example}
The LCTRS $\cRflist$ in Example~\ref{ex:flist} is not quasi-reducible
because
$
\{ \CTerm{\symb{f}(\var{xs},y)}{\symb{true}} \} \ominus \{ (1), (2), (3) \}
=
\{ (\mathrm{a}), (\mathrm{b}), (\mathrm{c}) \}
$.
On the other hand, the LCTRS $\cRflist'$ is quasi-reducible
because
$
\{ \CTerm{\symb{f}(\var{xs},y)}{\symb{true}} \} \ominus \{ (1), (2), (3), (\mathrm{a}), (\mathrm{b}), (\mathrm{c}) \}
=
\emptyset
$.
\end{example}


\section{Related Work}
\label{sec:related-work}

As mentioned before, to the best of our knowledge, there is no work for difference operators and complement algorithms for constrained patterns, and decidability of quasi-reducibility in the LCTRS setting, while some sufficient conditions have been investigated~\cite{SNSSK09,Kop17}.

There are several works on quasi-reducibility and \emph{sufficient completeness} of TRSs and constrained rewrite systems.
Decidability of quasi-reducibility has been shown in~\cite{KNZ87} without yielding a practical algorithm.
Then, complement algorithms have been proposed~\cite{LM86,LLT90,HA19pro}, which are used as decision procedures for quasi-reducibility of left-linear TRSs.
\emph{Negation elimination} from equational formulas~\cite{Taj93,Fer98} is a well-investigated application of complement algorithms.
\emph{Tree automata techniques}~\cite{tata2007} can be used to obtain complements,
and thus they also yield decision procedures for quasi-reducibility of left-linear TRSs.
Recently, a well-designed decision procedure for quasi-reducibility of TRSs has been proposed~\cite{TY24}.
However, to the best of our knowledge, the above procedures have not yet been extended to constrained rewrite systems, and the extension is not so trivial.

Sufficient completeness is equivalent to quasi-reducibility for terminating systems.
A sufficient condition for sufficient completeness of constrained and conditional rewrite systems based on constrained tree automata techniques has been shown in~\cite{BJ12}.
In addition, as mentioned in Section~\ref{sec:intro}, 
an SMT-based sufficient condition for left-patternless constrained rewrite systems~\cite{SNSSK09,KNM25jlamp} and a procedural sufficient condition for quasi-reducibility of LCTRSs~\cite{Kop17} have been shown.

For programming languages, there are some works on the exhaustiveness checking.
In \cite{Mar07}, a method for examining pattern-matching anomalies of ML programming languages has been proposed. 
This method checks two points: there are no useless patterns, and it considers all patterns. This method is implemented on OCaml but is also useful for Haskell. 
However, guards are not considered:
Standard ML of New Jersey does not allow us to write guards for pattern matching;
a recent version of the GHC implementation of Haskell does not implement any form of exhaustiveness checking even without guards, while there is a check for redundancy of patterns, which does not work for guards.
When compiling programs, OCaml compilers warn users about guard constraints attached to patterns so that they can take care of the exhaustiveness of constrained patterns, which is not checked by the compilers.
As mentioned before, due to the SMT solving for supported built-in theories, the exhaustiveness checking for guarded patterns is not realistic.


Unification of constrained terms in Definition~\ref{def:unifiability-of-constrained-terms} is a bit different from \emph{unification modulo built-ins}~\cite[Definition~11]{CAL18}.
Unlike ours, the unification is defined for terms, and unifiers are pairs of substitutions and constraints.

Another related work on unification of constrained terms is
unification with abstraction and theory instantiation for SMT solving~\cite{RSV18}.
The algorithm for such unification is very similar to the syntactic unification procedure~\cite[Section~4.6]{BN98}, while it considers underlying theories and disequalities as constraints.
Unlike the algorithm, our unification of constrained terms $\CTerm{s}{\phi},\CTerm{t}{\psi}$ first uses syntactic unification to obtain an mgu $\sigma$ of $s$ and $t$, and then check satisfiability of $\phi\sigma \land \psi\sigma$.

\emph{Matching logic}~\cite{Ros17,CR19lics,CR19calco} (ML, for short) is a first-order logic variant
to reason about structures by allowing formulas to represent terms, and is expressive enough and very powerful.
For example, $s \ominus t$ with $\Var(s) \cap \Var(t) = \emptyset$ can be represented by the formula $s \land \neg t$.
However, decidability of ML formulas is not clear due to its generality.








\section{Conclusion}
\label{sec:conclusion}

In this paper, we extended the difference operator and the complement algorithm in the unconstrained setting to constrained linear patterns, proposing a difference operator and a complement algorithm for constrained linear patterns used in LCTRSs that have finitely many user-defined function symbols and have no user-defined constructor term with a theory sort for built-in value-constants.
Then, we showed that quasi-reducibility is decidable for such LCTRSs.
Many LCTRSs obtained from practical programs by existing transformations belong to the class.
We will implement the algorithm in our LCTRS tool as future work.

For brevity, we had some restrictions, e.g., non-existence of user-defined constructor terms with theory sorts.
A future work is to drop such restrictions, proposing a difference operator and a complement algorithm for more general classes of LCTRSs.

The complexity of the complement algorithm for constrained linear patterns relies on that of SMT solving used in the computation.
We are interested in complexity of the algorithm without the complexity of SMT solving.
Analysis of such complexity is another future work.


%
%
%
\bibliographystyle{splncs04}
\bibliography{mybiblio}

\appendix
\input{appendix}

\end{document}

%% file: appendix.tex
\section{Auxiliary Claims for Missing Proofs}

Let $\Sigma' \subseteq \Sigma$.
For unifiable terms $s,t$, their $\Sigma'$-mgu $\sigma$ splits $\GInst[\Sigma']{s} \cup \GInst[\Sigma']{t}$ to the disjoint sets
$\GInst[\Sigma']{s\sigma}$, $\GInst[\Sigma']{s}\setminus \GInst[\Sigma']{s\sigma}$, and $\GInst[\Sigma']{t}\setminus \GInst[\Sigma']{t\sigma}$.
Note that $\GInst[\Sigma']{s\sigma} = \GInst[\Sigma']{t\sigma}$.

\begin{restatable}{proposition}{PropDisjointnessOfFUnifiableTerms}
\label{prop:disjointness-of-F-unifiable-terms}
Let $\Sigma' \subseteq \Sigma$, $s,t$ be terms such that $\Var(s)\cap\Var(t)=\emptyset$ and $s,t$ are $\Sigma'$-unifiable.
Let $\sigma = \mgu[\Sigma'](s,t)$.
Then, both of the following statements hold:
\begin{enumerate}
\renewcommand{\labelenumi}{(\arabic{enumi})}
    \item 
$\GInst[\Sigma']{s\sigma}$, $\GInst[\Sigma']{s}\setminus \GInst[\Sigma']{s\sigma}$, and $\GInst[\Sigma']{t}\setminus \GInst[\Sigma']{t\sigma}$ are  pairwise disjoint, 
    and
    \item
$
\GInst[\Sigma']{s} \cup \GInst[\Sigma']{t} = \GInst[\Sigma']{s\sigma} \uplus (\GInst[\Sigma']{s}\setminus \GInst[\Sigma']{s\sigma}) \uplus (\GInst[\Sigma']{t}\setminus \GInst[\Sigma']{t\sigma})
$.
\end{enumerate}
\end{restatable}
\begin{proof}
We first show the first claim~(1).
It is clear that $\GInst[\Sigma']{s\sigma}$ and $\GInst[\Sigma']{s}\setminus \GInst[\Sigma']{s\sigma}$ are disjoint and $\GInst[\Sigma']{s\sigma}$ and $\GInst[\Sigma']{t}\setminus \GInst[\Sigma']{t\sigma}$ are disjoint.
We show that $\GInst[\Sigma']{s}\setminus \GInst[\Sigma']{s\sigma}$ and $\GInst[\Sigma']{t}\setminus \GInst[\Sigma']{t\sigma}$ are disjoint.
We proceed by contradiction.
Assume that $(\GInst[\Sigma']{s}\setminus \GInst[\Sigma']{s\sigma}) \cap (\GInst[\Sigma']{t}\setminus \GInst[\Sigma']{t\sigma}) \ne \emptyset$.
Then, there exist $\Sigma'$-substitutions $\theta,\theta'$ such that $s\theta = t\theta' \in \GInst[\Sigma']{s} \cap \GInst[\Sigma']{t}$, $s\theta \notin \GInst[\Sigma']{s\sigma}$, and $t\theta' \notin \GInst[\Sigma']{t\sigma}$.
Since $s\theta = t\theta'$ and $\Var(s) \cap \Var(t) = \emptyset$, $(\theta \cup \theta')$ is a unifier of $s,t$ and thus $\sigma \lesssim \theta$ and $\sigma \lesssim \theta'$.
Thus, we have that $s\theta \in \GInst[\Sigma']{s\sigma}$ and $t\theta' \in \GInst[\Sigma']{t\sigma}$.
This contradicts the fact that $s\theta \notin \GInst[\Sigma']{s\sigma}$ and $t\theta' \notin \GInst[\Sigma']{t\sigma}$.

Next, we show the second claim~(2).
By definition, it is clear that $\GInst[\Sigma']{s} \supseteq \GInst[\Sigma']{s\sigma}$ and $\GInst[\Sigma']{t} \supseteq \GInst[\Sigma']{t\sigma}$ and hence 
$\GInst[\Sigma']{s} = \GInst[\Sigma']{s\sigma} \cup (\GInst[\Sigma']{s}\setminus \GInst[\Sigma']{s\sigma})$
and
$\GInst[\Sigma']{t} = \GInst[\Sigma']{t\sigma} \cup (\GInst[\Sigma']{t}\setminus \GInst[\Sigma']{t\sigma})$.
Since $s\sigma = t\sigma$, we have that $\GInst[\Sigma']{s\sigma} = \GInst[\Sigma']{t\sigma}$.
Therefore, by the first claim~(1), we have that
$
\GInst[\Sigma']{s} \cup \GInst[\Sigma']{t} = \GInst[\Sigma']{s\sigma} \uplus (\GInst[\Sigma']{s}\setminus \GInst[\Sigma']{s\sigma}) \uplus (\GInst[\Sigma']{t}\setminus \GInst[\Sigma']{t\sigma})
$.
\qed
\end{proof}

We have a sufficient condition for $\cC$-non-overlappingness of linear terms.
\begin{restatable}{proposition}{PropSufficientConditionForFDisjointnessOfTerms}
\label{prop:sufficient-condition-for-F-non-overlappingness-of-terms}
Let 
$s,t$ be terms such that $\Pos(s) \cap \Pos(t) \ne \emptyset$.
If $s|_p$ and $t|_p$ are $\cC$-non-overlapping for some $p \in \Pos(s) \cap \Pos(t)$, then $s$ and $t$ are $\cC$-non-overlapping.
\end{restatable}
\begin{proof}
We prove this claim by structural induction on $s$.
Let $p \in \Pos(s) \cap \Pos(t)$ such that $s|_p$ and $t|_p$ are $\cC$-non-overlapping.
Then, neither $s|_p$ nor $t|_p$ is a variable and hence neither $s$ nor $t$ is a variable.
Since the case where $p=\epsilon$ or $s$ and $t$ are rooted by different function symbols is trivial, we consider the remaining case where $p \ne \epsilon$ and $s$ and $t$ are rooted by the same function symbol.
Let $s = f(s_1,\ldots,s_n)$, $t = f(t_1,\ldots,t_n)$, and $p = ip'$ such that $1 \leq i \leq n$.
Note that $p' \in \Pos(s_i) \cap \Pos(t_i)$ by definition.
By the induction hypothesis, $s_i$ and $t_i$ are $\cC$-non-overlapping.
We now assume that $s$ and $t$ are not $\cC$-non-overlapping.
Then, by definition, there exist substitutions $\sigma,\sigma'$ such that $s\sigma = t\sigma' \in \GInst[\cC]{s} \cap \GInst[\cC]{t}$, and thus $s_i\sigma = t_i\sigma' \in \GInst[\Sigma']{s} \cap \GInst[\cC]{t}$.
Thus, $s_i$ and $t_i$ are not $\cC$-non-overlapping.
This contradicts the $\cC$-non-overlappingness of $s_i$ and $t_i$.
Therefore, $s$ and $t$ are $\cC$-non-overlapping.
\qed
\end{proof}

\section{Proof of Proposition~\ref{prop:sufficient-condition-of-linearity-of-mgu}}

To prove Proposition~\ref{prop:sufficient-condition-of-linearity-of-mgu}, we generalize it to unification problems $S$ of the form $\{ {s_1 =^? t_1},\ldots,{s_n =^? t_n} \}$ and use the unification procedure in~\cite[Section~4.6]{BN98} as follows:
\begin{itemize}
    \item For a term $s$ and a linear term $t$ with $\Var(s) \cap \Var(t) = \emptyset$, we consider the unification problem $\{ {s =^? t} \}$.
    \item The disjointness of variables ($\Var(s) \cap \Var(t) = \emptyset$) is formulated as 
    the \emph{LR-separation} of $S$: $\Var(s_1,\ldots,s_n) \cap \Var(t_1,\ldots,t_n)=\emptyset$.
    \item The linearity of $t$ is formulated as the \emph{R-linearity} of $S$:
    \begin{itemize}
        \item $t_1,\ldots,t_n$ are linear,
            and
        \item $\Var(t_i) \cap \Var(t_j) = \emptyset$ for all $1 \leq i < j \leq n$.
    \end{itemize}
    \item To preserve the two properties by transformation steps of unification problems, we do not use \textsf{Orient} that replaces $t =^? x$ by $x =^? t$, where $t$ is not a variable.
\end{itemize}

We first recall the unification problems.
\begin{definition}[unification problem~\cite{BN98}] 
A \emph{unification problem} $S$ is a finite set of equations $s =^? t$ over terms:
$S = \{ {s_1 =^? t_1}, \ldots, {s_n =^? t_n} \}$.
A \emph{unifier} or \emph{solution} of $S$ is a substitution $\sigma$ such that
$s_i\sigma = t_i\sigma$ for all $1 \leq i \leq n$.
$\cU(S)$ denotes the set of unifiers of $S$.
$S$ is called \emph{unifiable} if $\cU(S) \ne \emptyset$.
%
A substitution $\sigma$ is called a \emph{most general unifier} (mgu, for short) of $S$ if $\sigma$ is a least element of $\cU(S)$, i.e.,
$\sigma \in \cU(S)$ and $\sigma \lesssim \sigma'$ for all $\sigma' \in \cU(S)$.
\end{definition}
Note that $\sigma$ is an mgu of terms $s,t$ if and only if $\sigma$ is an mgu of the unification problem $\{ {s =^? t} \}$.
A substitution $\sigma$ is called \emph{idempotent} if $\sigma = (\sigma \circ \sigma)$ (i.e., $\Ran(\sigma) \cap \VRan(\sigma) = \emptyset$~\cite[Lemma~4.5.7]{BN98}).


\begin{theorem}[\cite{BN98}] 
If a unification problem $S$ has a solution, then $S$ has an idempotent mgu.
\end{theorem}

Our goal is to show that, given a term $s$ and a linear term $t$ with $\Var(s) \cap \Var(t) = \emptyset$, 
if $\{ {s =^? t} \}$ has an mgu $\sigma$ as as solution, then $x\sigma$ is linear for any variable $x \in \Var(s)$.
To distinguish variables in $s,t$, we distinguish $u =^? u'$ and $u' =^? u$
and we do not use the transformation step \textsf{Orient} which replaces $t =^? x$ by $x =^? t$, where $t$ is not a variable.
For this reason, we slightly extend solved forms~\cite[Definition~4.6.1]{BN98} of unification problems.
\begin{definition}[solved form]
A unification problem $S$ is said to be \emph{in solved form} if 
$S$ is of the following form
\[
\{ {x_1 =^? t_1}, \ldots, {x_m =^? t_m} \} \uplus \{ {t_{m+1} =^? x_{m+1}}, \ldots, {t_n =^? x_n} \}
\]
where
\begin{itemize}
    \item $0 \leq m \leq n$,
    \item none of $t_{m+1},\ldots,t_n$ is a variable,
    \item $x_1,\ldots,x_m$ are pairwise different variables, 
        and
    \item $\{x_1,\ldots,x_n\} \cap \Var(t_1,\ldots,t_n) = \emptyset$.
\end{itemize}
Note that $t_1,\ldots,t_m$ may be variables.
We denote $\{ x_i \mapsto t_i \mid 1 \leq i \leq n\}$ by $\vec{S}$.
\end{definition}


\begin{lemma}[\cite{BN98}] 
If a unification problem $S$ is in solved form, then $\vec{S}$ is an idempotent mgu of $S$.
\end{lemma}

Because of our solved form, by introducing a variant of \textsf{Eliminate}, we do not have to use the transformation step \textsf{Orient} in order to solve given unification problems.
\begin{definition}[unification transformation $\UnifStep$]
The transformation of unification problems is defined as a binary relation ${\UnifStep} = ({\UnifStep[Delete]} \cup {\UnifStep[Decompose]} \cup {\UnifStep[EliminateL]} \cup {\UnifStep[EliminateR]})$, where 
\begin{itemize}
    \item $\{{ t =^? t} \} \uplus S \UnifStep[Delete] S$,
    \item $\{ {f(s_1,\ldots,s_n) =^? f(t_1,\ldots,t_n)} \} \uplus S \UnifStep[Decompose] \{ {s_1 =^? t_1}, \ldots, {s_n =^? t_n} \} \cup S$,
    \item $\{ {x =^? t} \} \uplus S \UnifStep[EliminateL] \{ {x =^? t} \} \cup (\{x\mapsto t\}(S))$,
    if $x \in \Var(S) \setminus \Var(t)$,
        and
    \item $\{ {t =^? x} \} \uplus S \UnifStep[EliminateR] \{ {t =^? x} \} \cup (\{x\mapsto t\}(S))$,
    if $x \in \Var(S) \setminus \Var(t)$.
\end{itemize}
\end{definition}
In the original framework, for equations of the form $t =^? x$ with $t \notin \cV$, we apply \textsf{Orient} to $\{ {t =^? x} \} \uplus S$ and then apply \textsf{EliminateL} to $\{ {x =^? t} \} \cup S$ if necessary.
This role is done by \textsf{EliminateR} in our modified transformation.
Thus, our transformation has the same properties of the original one.
\begin{lemma}[\cite{BN98}] 
All of the following statements hold:
\begin{enumerate}
    \item The unification transformation $\UnifStep$ is well-founded.
    \item For unification problems $S,S'$, if $S \UnifStep S'$, then $\cU(S) = \cU(S')$.
    \item If a unification problem $S$ is solvable, then there exists a unification problem $S'$ in solved form such that $S \UnifStep^* S'$.
\end{enumerate}
\end{lemma}


The application of a substitution $\sigma$ to a unification problem $S$ 
is denoted by $\sigma(S)$:
$\sigma(S) =  \{ {s\sigma =^? t\sigma} \mid (s =^? t) \in S \}$.
The set of variables appearing in $S$ 
is denoted by $\Var(S)$:
$\Var(S) = \bigcup_{(s =^? t) \in S} \Var(s,t)$.
The set of variables appearing in the left-hand sides of equations in $S$ is denoted by $\Var_L(S)$:
$\Var_L(S) = \bigcup_{(s =^? t) \in S} \Var(s)$.
The set of variables appearing in the right-hand sides of equations in $S$ is denoted by $\Var_R(S)$:
$\Var_R(S) = \bigcup_{(s =^? t) \in S} \Var(t)$.
Note that $\Var(S) = \Var_L(S) \cup \Var_R(S)$.

An equation $x =^? t$ in a unification problem $S$ is called \emph{L-solved in S} if $x \notin \Var(S \setminus \{ {x =^? t} \}) \cup \Var(t)$.
An equation $t =^? x$ in $S$ is called \emph{R-solved in S} if $x \notin \Var(S \setminus \{ {t =^? x} \}) \cup \Var(t)$ and $t \notin \cV$.
Note that there is no L-solved and R-solved equation in $S$.
An equation in $S$ is called \emph{solved in $S$} if it is L-solved or R-solved in $S$.
An equation in $S$ is called \emph{unsolved in $S$} if it is not solved in $S$ (i.e., it is neither L-solved nor R-solved in $S$).


We call a unification problem $S$ \emph{LR-separated} if $\Var_L(S) \cap \Var_R(S) = \emptyset$.
We call a unification problem $\{ {s_1 =^? t_1}, \ldots, {s_n =^? t_n} \}$ \emph{R-linear} if 
    $t_1,\ldots,t_n$ are linear
        and
    $\Var(t_i) \cap \Var(t_j) = \emptyset$ for all $1 \leq i < j \leq n$.

The transformation step $\UnifStep[EliminateR]$ is not applicable to any LR-separated and R-linear unification problem.
\begin{lemma}
\label{lem:EliminateR-is-not-applicable}
If a unification problem $S$ is LR-separated and R-linear, then $S \not\UnifStep[EliminateR] {}$.
\end{lemma}
\begin{proof}
We proceed by contradiction. 
Assume that 
\[
\begin{array}{@{}r@{\>}c@{\>}l}
S & = & \{ {t =^? x} \} \uplus \{ {s_1 =^? t_1}, \ldots, {s_n =^? t_n} \} \\
&& {} \UnifStep[EliminateR] \{ {t =^? x} \} \cup \{ {s_1\sigma =^? t_1\sigma}, \ldots, {s_n\sigma =^? t_n\sigma} \} = S'
\end{array}
\]
where $\sigma = \{ x \mapsto t \}$
and $x \in \Var(s_1,t_1,\ldots,s_n,t_n) \setminus \Var(t)$.
Since $S$ is LR-separated and R-linear, we have that
\begin{itemize}
    \item $\Var_L(S) \cap \Var_R(S)=\emptyset$ (i.e., $\Var_L(t,s_1,\ldots,s_n) \cap \Var_R(x,t_1,\ldots,t_n) = \emptyset$),
    \item $t_0,t_1,\ldots,t_n$ are linear,
    \item $x \notin \Var(t_1,\ldots,t_n)$,
        and
    \item $\Var(t_i) \cap \Var(t_j) = \emptyset$ for all $1 \leq i < j \leq n$
\end{itemize}
and hence
$x \notin \Var(t,s_1,t_1,\ldots,s_n,t_n)$.
This contradicts the assumption that $x \in \Var(s_1,t_1,\ldots,s_n,t_n) \setminus \Var(t)$.
\qed
\end{proof}

In the following, by Lemma~\ref{lem:EliminateR-is-not-applicable}, we do not consider the transformation step $\UnifStep[EliminateR]$ for LR-separated and R-linear unification problems.

\begin{lemma}
\label{lem:LR-separated-and-R-linear-preservation}
Let $S,S'$ be unification problems such that $S \UnifStep S'$.
If $S$ is LR-separated and R-linear, then 
\begin{itemize}
    \item $S'$ is LR-separated and R-linear,
        and
    \item $\Var_L(S) \cap \Var_R(S') = \emptyset$.
\end{itemize}
\end{lemma}
\begin{proof}
Let $S = \{ {s_0 =^? t_0}, {s_1 =^? t_1}, \ldots, {s_n =^? t_n} \}$ for some $n \geq 0$.
Assume that $s_0 =^? t_0$ is a focused equation at the transformation step $S \UnifStep S'$.
We make a case analysis depending on which rule of $\UnifStep$ is applied to $S$ to obtain $S'$.
\begin{itemize}
    \item Case where $\UnifStep[Delete]$ is applied to $S$.
    Assume that $s_0 = t_0$ and 
    \[
    S = \{ {t_0 =^? t_0}, {s_1 =^? t_1}, \ldots, {s_n =^? t_n} \} \UnifStep[Delete] \{ {s_1 =^? t_1}, \ldots, {s_n =^? t_n} \} = S'
    \]
    Since $S$ is LR-separated and R-linear, we have that
    \begin{itemize}
        \item $\Var_L(S) \cap \Var_R(S) = \emptyset$,
        \item $t_0,t_1,\ldots,t_n$ are linear,
            and
        \item $\Var(t_i) \cap \Var(t_j) = \emptyset$ for all $0 \leq i < j \leq n$,
    \end{itemize}
    and thus 
    \begin{itemize}
        \item $\Var(t_0) = \emptyset$,
        \item $\Var(t_0,s_1,\ldots,s_n) \cap \Var(t_0,t_1,\ldots,t_n) = \emptyset$,
        \item $t_1,\ldots,t_n$ are linear,
            and
        \item $\Var(t_i) \cap \Var(t_j) = \emptyset$ for all $1 \leq i < j \leq n$.
    \end{itemize}
    Since $\Var(t)=\emptyset$, we have that $\Var_L(S) = \Var(s_1,\ldots,s_n)$, and thus 
    we have that $\Var_L(S) \cap \Var_R(S') = \emptyset$.
    Therefore, $S'$ is also LR-separated and R-linear, and $\Var_L(S) \cap \Var_R(S') = \emptyset$.

    \item Case where $\UnifStep[Decompose]$ is applied to $S$.
    Assume that $s_0 = f(s'_1,\ldots,s'_m)$, $t_0 = f(t'_1,\ldots,t'_m)$, and 
    \[
    \begin{array}{@{}r@{\>}c@{\>}l}
    S & = & \{ {f(s'_1,\ldots,s'_m) =^? f(t'_1,\ldots,t'_m)} \} \uplus \{ {s_1 =^? t_1}, \ldots, {s_n =^? t_n} \} \\
    && {} \UnifStep[Decompose] \{ {s'_1 =^? t'_1},\ldots,{s'_m =^? t'_m} \} \cup \{ {s_1 =^? t_1}, \ldots, {s_n =^? t_n} \} = S'
    \end{array}
    \]
    Since $S$ is LR-separated and R-linear, we have that
    \begin{itemize}
        \item $\Var_L(S) \cap \Var_R(S) = \emptyset$,
        \item $t_0,t_1,\ldots,t_n$ are linear,
            and
        \item $\Var(t_i) \cap \Var(t_j) = \emptyset$ for all $0 \leq i < j \leq n$
    \end{itemize}
    and hence
    \begin{itemize}
        \item $\Var(s'_1,\ldots,s'_m,s_1,\ldots,s_n) \cap \Var(t'_1,\ldots,t'_m,t_1,\ldots,t_n) = \emptyset$,
        \item $t'_1,\ldots,t'_m,t_1,\ldots,t_n$ are linear,
        \item $\Var(t'_i) \cap \Var(t'_j) = \emptyset$ for all $1 \leq i < j \leq m$,
        \item $\Var(t_i) \cap \Var(t_j) = \emptyset$ for all $1 \leq i < j \leq n$,
            and
        \item $\Var(t'_i) \cap \Var(t_j) = \emptyset$ for any $i \in \{1,\ldots,m\}$ and $j \in \{1,\ldots,n\}$.
    \end{itemize}
    By definition, we have that 
    \begin{itemize}
        \item $\Var_L(S) = \Var(s_0,s_1,\ldots,s_n)=\Var(s'_1,\ldots,s'_m,s_1,\ldots,s_n)$,
            and
        \item $\Var_R(S') = \Var(t_0,t_1,\ldots,t_n) = \Var(t'_1,\ldots,t'_m,t_1,\ldots,t_n)$.
    \end{itemize}
    Therefore, $S'$ is LR-separated and R-linear, and $\Var_L(S) \cap \Var_R(S') = \emptyset$.

    \item Case where $\UnifStep[EliminateL]$ is applied to $S$.
    Assume that $s_0$ is a variable $x \in \Var(s_1,t_1,\ldots, s_n,t_n) \setminus \Var(t_0)$
    (i.e., $s_0 = x \in \Var(s_1,t_1,\ldots, s_n,t_n) \setminus \Var(t_0)$) and
    \[
    \begin{array}{@{}r@{\>}c@{\>}l}
    S & = & \{ {x =^? t_0} \} \uplus \{ {s_1 =^? t_1}, \ldots, {s_n =^? t_n} \} \\
    && {} \UnifStep[EliminateL] \{ {x =^? t_0} \} \cup \{ {s_1\sigma =^? t_1\sigma}, \ldots, {s_n\sigma =^? t_n\sigma} \} = S'
    \end{array}
    \]
    where $\sigma = \{ x \mapsto t \}$.
    Since $S$ is LR-separated and R-linear, we have that
    \begin{itemize}
        \item $\Var_L(S) \cap \Var_R(S) = \emptyset$,
        \item $t_0,t_1,\ldots,t_n$ are linear,
            and
        \item $\Var(t_i) \cap \Var(t_j) = \emptyset$ for all $0 \leq i < j \leq n$
    \end{itemize}
    and hence
    $\Var_L(x,s_1,\ldots,s_n) \cap \Var_R(t_0,t_1,\ldots,t_n) = \emptyset$.
    Thus, we have that
    \begin{itemize}
        \item $(s_i\sigma =^? t_i\sigma) = (s_i\sigma =^? t_i)$ for all $1 \leq i \leq n$,
            and
        \item  $\Var_L(x,s_1,\ldots,s_n) \cap \Var_R(t_1,\ldots,t_n) = \emptyset$.
    \end{itemize}
    By definition, we have that 
    $\Var(s_i\sigma) = (\Var(s_i)\setminus\{x\}) \cup \Var(t_0)$,
    and thus
    $\Var(s_1\sigma,\ldots,s_n\sigma) = (\Var(s_1,\ldots,s_n) \setminus \{x\}) \cup \Var(t_0)
    = \Var(s_1,\ldots,s_n,t_0) \setminus \{x\}$.
    Since $\Var(t_0) \cap \Var(t_j) = \emptyset$ for all $1 \leq j \leq n$,
    we have that $\Var(t_0) \cap \Var(t_1,\ldots,t_n) = \emptyset$ and hence
    $\Var(s_1,\ldots,s_n,t_0) \cap \Var(t_1,\ldots,t_n)=\emptyset$.
    Thus, we have that $\Var(s_1\sigma,\ldots,s_n\sigma) \cap \Var(t_1,\ldots,t_n) = \emptyset$.
    Therefore, $S'$ is LR-separated and R-linear, and $\Var_L(S) \cap \Var_R(S') = \emptyset$.
\qed
\end{itemize}
\end{proof}

\begin{lemma}
\label{lem:LR-separated-and-R-linear-preservation-by-multi-step}
Let $S,S'$ be unification problems such that $S \UnifStep^* S'$.
If $S$ is LR-separated and R-linear, then 
\begin{itemize}
    \item $S'$ is LR-separated and R-linear,
        and
    \item $\Var_L(S) \cap \Var_R(S') = \emptyset$.
\end{itemize}
\end{lemma}
\begin{proof}
We prove this claim by induction on the length of $S \UnifStep^* S'$.
Assume that $S \UnifStep^n S'$.
Since the case where $n=0$ (i.e., $S = S'$) is trivial, we consider the remaining case where $n > 0$.
Suppose that $S \UnifStep S'' \UnifStep^{n-1} S'$.
It follows from Lemma~\ref{lem:LR-separated-and-R-linear-preservation} that 
$S''$ is LR-separated and R-linear, and $\Var_L(S) \cap \Var_R(S'') = \emptyset$.
By the induction hypothesis, $S'$ is LR-separated and R-linear, and $\Var_L(S'') \cap \Var_R(S') = \emptyset$.
We show that $\Var_L(S) \cap \Var_R(S') = \emptyset$, by contradiction.
Assume that there exists a variable $x \in \Var_L(S) \cap \Var_R(S')$.
By definition, it is clear that $\Var(S) \supseteq \Var(S'') \supseteq \Var(S')$.
Thus, we have that $x \in \Var(S'')$.
Since $\Var_L(S'') \cap \Var_R(S') = \emptyset$, we have that $x \in \Var_R(S'')$ and hence $x \in \Var_L(S) \cap \Var_R(S'')$.
This contradicts the fact that $\Var_L(S) \cap \Var_R(S'') = \emptyset$.
\qed
\end{proof}

\PropSufficientConditionOfLinearityOfMgu*
\begin{proof}
Assume that $t$ is linear.
We consider the unification problem $S = \{ {s =^? t} \}$.
Since $t$ is linear and $\Var(s)\cap\Var(t) = \emptyset$, the problem $S$ is LR-separated and R-linear.
Since $s,t$ are unifiable, there exists a unification problem $S'$ in solved form such that
$S \UnifStep^* S'$.
It follows from Lemma~\ref{lem:LR-separated-and-R-linear-preservation-by-multi-step} that 
$S'$ is LR-separated and R-linear, and $\Var_L(S) \cap \Var_R(S') = \emptyset$, and thus $\Var(s) \cap \Var_R(S') = \emptyset$.
Thus, for any variable $x \in \Var(s)$, one of the following holds:
\begin{itemize}
    \item $x$ does not appear in $S'$ (i.e., $x\vec{S'} = x$), 
        or
    \item there exists an equation $x =^? t$ in $S'$ such that $t$ is linear (and thus, $x\vec{S'} = t$).
\end{itemize}
In both cases above, $x\vec{S'}$ is linear.
Since $\vec{S'}$ is an mgu of $S$ (i.e., an mgu of $s,t$), $x\sigma$ is linear for any mgu $\sigma$ of $s,t$.
\qed
\end{proof}

\section{Proof of Theorem~\ref{thm:constrained-pattern-ominus-propertiesX}}

The following proposition shows correctness of $\Cocterm[\cC]{\cdot}$ for linear terms, together with some properties.
%
\begin{restatable}[correctness of {$\Cocterm[\cC]{u}$}]{proposition}{PropCoctermsCorrectness}
\label{prop:Cocterms-correctness}
Let 
$u:\iota$ be a linear $\cC$-term.
Then, $\Cocterm[\cC]{u}$ is a finite pairwise $\cC$-non-overlapping complement of $u$
(i.e., $\{ t \in T(\cC) \mid t:\iota\} = \GInst[\cC]{u} \uplus \GInst[\cC]{\Cocterm[\cC]{u}}$)
such that
\begin{enumerate}
    \item any term in $\Cocterm[\cC]{u}$ is a linear non-variable $\cC$-term with sort $\iota$,
    \item $u$ and $u'$ are $\cC$-non-overlapping for any $u' \in \Cocterm[\cC]{u}$,
        and
    \item $\Height{u} \geq \Height{u'}$ for any term $u' \in \Cocterm[\cC]{u}$.
\end{enumerate}
\end{restatable}
\begin{proof}
We prove this claim by structural induction on $u$.
Since the case where $u$ is a variable is trivial, we only show the remaining case where $u$ is not a variable.
Let $u = f(u_1,\ldots,u_n)$ such that $f : \iota_i \times \cdots \times \iota_n \to \iota \in \Sigma'$.
Let $\Cocterm[\Sigma']{u} = U_0 \cup \bigcup_{i=1}^n U_i$, where
\begin{itemize}
    \item $U_0 = \{ f'(x_1,\ldots,x_m) \mid f' \ne f, ~ f' : \iota'_1 \times \cdots \times \iota'_m \to \iota \in \Sigma' \}$,
        and
    \item $U_i = \{ f(u_1,\ldots,u_{i-1},u'_i,y_{i+1},\ldots,y_n) \mid u'_i \in \Cocterm[\Sigma']{u_i} \}$.
\end{itemize}
Note that by definition, $x_1,\ldots,x_m,y_2,\ldots,y_n$ are pairwise distinct fresh variables.
By the induction hypothesis, for any $i \in \{1,\ldots,n\}$,
\begin{enumerate}
\renewcommand{\labelenumi}{(\roman{enumi})}
\leftskip=1ex
    \item $\Cocterm[\Sigma']{u_i}$ is a finite pairwise $\cC$-non-overlapping set of linear non-variable $\cC$-terms with sort $\iota_i$, 
    \item $u_i$ and $u'_i$ are $\cC$-non-overlapping for any term $u'_i \in \Cocterm[\Sigma']{u_i}$,
    \item $\Height{u_i} \geq \Height{u'_i}$ for any term $u'_i \in \Cocterm[\Sigma']{u_i}$,
        and
    \item $\{ t_i \in T(\cC) \mid t_i:\iota_i\} = \GInst[\cC]{u_i} \uplus \GInst[\cC]{\Cocterm[\cC]{u'_i}}$.
\end{enumerate}

We first show that
$\Cocterm[\cC]{u}$ is a finite pairwise $\cC$-non-overlapping set of linear non-variable $\cC$-terms with sort $\iota$.
It suffices to show that
\begin{enumerate}
\renewcommand{\labelenumi}{(\alph{enumi})}
    \item $U_0$ is a finite pairwise $\cC$-non-overlapping set of linear non-variable $\cC$-terms with sort $\iota$,
    \item $U_i$ is a finite pairwise $\cC$-non-overlapping set of linear non-variable $\cC$-terms with sort $\iota_i$ for any $i \in \{1,\ldots,n\}$, 
    \item $U_0$ and $U_i$ are $\cC$-non-overlapping for any $i \in \{1,\ldots,n\}$,
        and
    \item $U_i$ and $U_j$ are $\cC$-non-overlapping for any $i,j \in \{1,\ldots,n\}$ with $i<j$.
\end{enumerate}
Since $\cC$ is finite, by definition, the substatement~(a) above trivially holds.
We show the remaining substatements~(b)--(d) in order.
\begin{enumerate}
\renewcommand{\labelenumii}{\alph{enumii}.}
\stepcounter{enumii}
    \item 
    Let $i \in \{1,\ldots,n\}$.
    By~(i) and the definition of $U_i$, all of the following hold:
    \begin{itemize}
        \item $U_i$ is finite,
        \item for any term $u'_i \in \Cocterm[\cC]{u_i}$, $f(u_1,\ldots,u_{i-1},u'_i,y_{i+1},\ldots,y_n)$ is a linear non-variable $\cC$-term with sort $\iota$,
        \item for any terms $u'_i,u''_i \in \Cocterm[\cC]{u_i}$, $u'_i$ and $u''_i$ are $\cC$-non-overlapping and hence
            $f(u_1,\ldots,u_{i-1},u'_i, y_{i+1},\ldots,y_n)$
            and
            $f(u_1,\ldots,u_{i-1},u''_i,y_{i+1},\ldots,y_n)$
            are $\cC$-non-overlapping.
    \end{itemize}
    Therefore, $U_i$ is a finite pairwise $\cC$-non-overlapping set of linear non-variable $\cC$-terms with sort $\iota_i$.

    \item 
    Let $i \in \{1,\ldots,n\}$.
    Any term in $U_0$ is rooted by a function symbol $f' \in \cC$ such that $f' \ne f$.
    Therefore, $U_0$ and $U_i$ are $\cC$-non-overlapping.
    
    \item 
    Let $i,j \in \{1,\ldots,n\}$ with $i < j$, $u_i \notin \cV$, and $u_j \notin \cV$.
    Then, we have that $\Cocterm[\Sigma']{u_i} \ne \emptyset$ and $\Cocterm[\Sigma']{u_j} \ne \emptyset$.
    Let $u'_i\in \Cocterm[\Sigma']{u_i}, u'_j \in \Cocterm[\Sigma']{u_j}$.
    Since $u_i$ and $u'_i$ are $\cC$-non-overlapping,
    $f(u_1,\ldots,u_{i-1},u'_i,y_{i+1},\ldots,y_n)$
    and
    $f(u_1,\ldots,u_{j-1},u'_j,y'_{j+1}, \ldots,y'_n)$ are also $\cC$-non-overlapping.
    Therefore, $U_i,U_j$ are $\cC$-non-overlapping.
    
\end{enumerate}

Secondly, we show the statement~1.
It suffices to show that
for any $i \in \{0.\ldots,n\}$,
$u$ and $u'$ are $\cC$-non-overlapping for any $u' \in U_i$.
We make a case analysis depending on what $i$ is.
\begin{itemize}
    \item Case where $i=0$.
    Let $u' \in U_0$.
    Then, by definition, $u'$ is rooted by a function symbol $f' \in \Sigma'$ such that $f'\ne f$.
    Therefore, $u$ and $u'$ are $\cC$-non-overlapping.
    \item Case where $i \in \{1,\ldots,n\}$.
    Let $u' = f(u_1,\ldots,u_{i-1},u'_i,y_{i+1},\ldots,y_n)$ for some $u'_i \in \Cocterm[\Sigma']{u_i}$.
    It follows from~(ii) and Proposition~\ref{prop:sufficient-condition-for-F-non-overlappingness-of-terms} that $u$ and $u'$ are $\cC$-non-overlapping.
\end{itemize}

Thirdly, we show the statement~2.
Since $u' \in \Cocterm[\Sigma']{u}$, we have that $u' \in U_i$ for some $i \in \{0,\ldots,n\}$.
%
    By definition, we have that $\Height{u}=\Height{f(u_1,\ldots,u_n)} = 1 + \max\{\Height{u_1},\ldots,\Height{u_n}\} \geq 1$.
    We make a case analysis depending on whether $u' \in U_0$.
    \begin{itemize}
        \item Case where $u' \in U_0$.
        In this case, $u'$ is of the form $f'(x_1,\ldots,x_m)$ for some $f' \in \cC$, and thus $\Height{u'} = 1$.
        Therefore, we have that $\Height{u} \geq 1 = \Height{u'}$.

        \item Case where $u' \notin U_0$.
        In this case, we have that $u' \in U_i$ for some $i \in \{1,\ldots,n\}$.
        Thus, $u'$ is of the form $f(u_1,\ldots,u_{i-1},u'_i,y_{i+1},\ldots,y_n)$ for some term $u'_i \in \Cocterm[\Sigma']{u_i}$.
        By definition, we have that 
        \[
        \begin{array}{@{}l@{\>}c@{\>}l@{}}
        \Height{u'} & = & \Height{f(u_1,\ldots,u_{i-1},u'_i,y_{i+1},\ldots,y_n)} \\
        & = & 1 + \max\{\Height{u_1},\ldots,\Height{u_{i-1}},\Height{u'_i}\} \\
        \end{array}
        \]
        Therefore, it follows from~(iii) that 
        $\Height{u} = 1 + \max\{\Height{u_j}  \mid 1 \leq j \leq n \} \geq 1+\max\{\Height{u_1},\ldots,\Height{u_{i-1}}, \Height{u'_i}\} = \Height{u'}$.
    \end{itemize}

Finally, we show that 
$\GInst[\cC]{\Cocterm[\cC]{u}}$ is a complement of $u$, 
i.e., $\{ t \in T(\cC) \mid t:\iota\} = \GInst[\cC]{u} \uplus \GInst[\cC]{\Cocterm[\cC]{u}}$.
It follows from~(2) that 
$\GInst[\cC]{u}$ and $\GInst[\cC]{\Cocterm[\cC]{u}}$ are disjoint (i.e., $\GInst[\cC]{u} \cap \GInst[\cC]{\Cocterm[\cC]{u}} = \emptyset$).
To prove this claim, it suffices to show that
any term $t \in T(\cC)$ with sort $\iota$ is in $\GInst[\cC]{u} \cup \GInst[\cC]{\Cocterm[\cC]{u}}$.
We proceed by contradiction.
Assume that there exists a term $t \in T(\cC)$ with sort $\iota$ such that $t \notin \GInst[\cC]{u} \cup \GInst[\cC]{\Cocterm[\cC]{u}}$.
Let $t = g(t_1,\ldots,t_k)$ such that $g: \iota'_1 \times \cdots \times \iota'_k \to \iota \in \cC$.
We make a case analysis depending on whether $g = f$.
\begin{itemize}
    \item Case where $g = f$.
    Then, there exists $i \in \{1,\ldots,n\}$ such that 
    \begin{itemize}
        \item $t_j \in \GInst[\cC]{u_j}$ for $j \in \{1,\ldots,i-1\}$
            and
        \item $t_i \notin \GInst[\cC]{u_i}$.
    \end{itemize}
    It follows from~(iv) that $t_i \in \GInst[\cC]{u_i} \uplus \GInst[\cC]{\Cocterm[\cC]{u'_i}}$ and hence $t_i \in \GInst[\cC]{\Cocterm[\cC]{u'_i}}$.
    Thus, we have that $t \in \GInst[\cC]{U_i}$.

    \item Case where $g \ne f$.
    In this case, by definition, we have that $t \in \GInst[\cC]{U_0}$.
\end{itemize}
In both cases above, there exists $i \in \{0,\ldots,n\}$ such that $t \in \GInst[\cC]{U_i}$.
This contradicts the assumption that  $t \notin \GInst[\cC]{u} \cup \GInst[\cC]{\Cocterm[\cC]{u}}$.
\qed
\end{proof}

In computing complements of linear patterns, we apply substitutions to linear patterns to obtain their linear instances.
We formulate substitutions preserving linearity in applying them to linear terms.

\begin{definition}[linearity-preserving substitution]
\label{def:linearity-preserving-substitution}
We say that a substitution $\sigma$ is \emph{linearity-preserving w.r.t.\ a set $X$ of variables} ($X$-linearity-preserving, for short) if
both of the following statements hold:
\begin{itemize}
   \item 
    $\Ran(\sigma|_X)$ is a set of linear terms,
        and
   \item
    $\Var(x\sigma) \cap \Var(y\sigma) = \emptyset$ for any variables $x,y \in X$ such that $x \ne y$.
\end{itemize}
\end{definition}
\begin{example}
The substitution $\{ \var{xs} \,{\mapsto}\, \symb{cons}(x',\var{xs}'), y \,{\mapsto}\, \symb{s}(y') \}$ is $\{\var{xs},y\}$-linearity-preserving but not $\{ \var{xs}, x' \}$-linearity-preserving.
\end{example}

The following proposition shows that any $X$-linearity-preserving substitution actually preserves linearity of terms $t$ with $\Var(t) \subseteq X$ and also shows a sufficient condition for $X$-linearity-preservation.
\begin{restatable}{proposition}{PropLinearSubstProperties}
\label{prop:linear-subst-properties}
Let $X \subseteq \cV$ and $\sigma$ be a substitution.
Then, both of the following statements hold:
\begin{enumerate}
\renewcommand{\labelenumi}{(\arabic{enumi})}
    \item
if $\sigma$ is $X$-linearity-preserving, then
        $t\sigma$ is linear for any linear term $t$ such that $\Var(t) \subseteq X$,
        and
    \item
if $t\sigma$ is linear for some term $t$ with $\Var(t) = X$, then
        $\sigma$ is $X$-linearity-preserving.
\end{enumerate}
\end{restatable}
\begin{proof}\ 
\begin{enumerate}
\renewcommand{\labelenumi}{(\arabic{enumi})}
    \item We proceed by contradiction.
    Assume that 
    $\sigma$ is $X$-linearity-preserving
    and
    there exists a linear term $t$ such that $\Var(t) \subseteq X$ and  $t\sigma$ is not linear.
    Then, there exist two parallel positions $p_1,p_2$ such that 
    $(t\sigma)|_{p_1} = (t\sigma)|_{p_2} \in \cV$.
    Since $(t\sigma)|_{p_1} \in \cV$ and $(t\sigma)|_{p_2} \in \cV$,
    there exist variables $x_1,x_2$ at positions $q_1,q_2$ of $t$, respectively, such that $q_1 \leq p_1$ and $q_2 \leq p_2$.
    Since $\sigma$ is $X$-linearity-preserving, we have that $x_1 \ne x_2$:
    If $x_1 = x_2$, then $(t\sigma)|_{p_1},(t\sigma)|_{p_2} \in \Var(x_1\sigma)$, and thus $x_1\sigma$ is not linear;
    this contradicts $X$-linearity-preservation of $\sigma$.
    Since $(t\sigma)|_{p_1} = (t\sigma)|_{p_2}$,
    we have that $(t\sigma)|_{p_1} \in \Var(x_1\sigma) \cap \Var(x_2\sigma)$.
    This contradicts $X$-linearity-preservation of $\sigma$.

    \item We proceed by contradiction.
    Let $t$ be a term such that $\Var(t)=X$.
    Assume that 
    $\sigma$ is linear and 
    $\sigma$ is not $X$-linearity-preserving.
    Then, we have that 
    \begin{itemize}
        \item there exists a variable $x \in \Var(t)$ such that $x\sigma$ is not linear,
            or
        \item there exist two different variables $x,x' \in \Var(t)$ such that $\Var(x\sigma) \cap \Var(x'\sigma) \ne \emptyset$.
    \end{itemize}
    In both cases above, $t\sigma$ is not linear.
    This contradicts the assumption that $t\sigma$ is linear.    
\qed
\end{enumerate}
\end{proof}

Linearity-preserving mgus have the following properties.

\begin{restatable}{proposition}{PropPropetiesOfFUnifiableLinearTerms}
\label{prop:propeties-of-F-unifiable-linear-terms}
Let 
$s,t$ be linear terms such that $\Var(s) \cap \Var(t) = \emptyset$
and $s$ and $t$ are unifiable.
Let $\sigma = \mgu(s,t)$.
Then, all of the following statements hold:
\begin{enumerate}
\renewcommand{\labelenumi}{(\arabic{enumi})}
    \item $\sigma$ is $\Var(s,t)$-linearity preserving,
    \item $\Pos(s) \cup \Pos(t) = \Pos(s\sigma)$ ($=\Pos(t\sigma)$),
    \item $\max\{\Height{s},\Height{t}\} = \max\{\Height{s\sigma},\Height{t\sigma}\}$,
        and
    \item if $s,t\in T(\cC,\cV)$ or both $s,t$ are patterns, then $\sigma$ is a $\cC$-substitution.
\end{enumerate}
\end{restatable}
\begin{proof}
Let $\Pos_\cV(s) = \{p_1,\ldots,p_n\}$ and $\Pos_\cV(t) \setminus \Pos_\cV(s) = \{ q_1,\ldots,q_m \}$.
Since $s,t$ are linear and $\Var(s) \cap \Var(t) = \emptyset$, 
$\{ s|_{p_i} \mapsto t|_{p_i} \mid 1 \leq i \leq n \} \cup \{ t|_{q_j} \mapsto s|_{q_j} \mid 1 \leq j \leq m \}$ is an mgu of $s,t$.
Thus, we let $\sigma = \{ s|_{p_i} \mapsto t|_{p_i} \mid 1 \leq i \leq n \} \cup \{ t|_{q_j} \mapsto s|_{q_j} \mid 1 \leq j \leq m \}$.
\begin{enumerate}
\renewcommand{\labelenumi}{(\arabic{enumi})}
    \item Since $\sigma$ is an mgu of $s,t$, by Proposition~\ref{prop:linear-subst-properties}~(2), it suffices to show that $s\sigma$ is linear.
    We proceed by contradiction.
    Assume that $s\sigma$ is not linear.
    Then, there exist two different variables positions $p,p' \in \Pos_\cV(s\sigma)$ such that $(s\sigma)|_p = (s\sigma)|_{p'}$.
    We make a case analysis depending on where $p,p'$ are.
    \begin{itemize}
        \item Case where $p,p' \in \Pos(s)$.
        Since $p,p' \in \Pos_\cV(s\sigma)$, we have that $p,p' \in \Pos_\cV(s)$, and thus $(s\sigma)|_p = t|_p$ and $(s\sigma)|_{p'}=t|_{p'}$.
        This contradicts the linearity of $t$.

        \item Case where $p \in \Pos(s)$ and $p'\ \notin \Pos(s)$.
        Then, there exists a position $q' \in \Pos_\cV(s)$ such that $q' \leq p'$.
        It follows from the construction of $\sigma$ that $(s\sigma)|_p= t|_p$ and $(s\sigma)|_{p'} = t|_{p'}$.
        This contradicts the linearity of $t$.

        \item Case where $p \notin \Pos(s)$ and $p' \in \Pos(s)$.
        This case can be proved similarly to the previous case. 

        \item Case where $p \notin \Pos(s)$ and $p' \notin \Pos(s)$.
        In this case, there are two positions $q,q' \in \Pos_\cV(s)$ such that $q \leq p$ and $q' \leq p'$.
        It follows from the construction of $\sigma$ that $(s\sigma)|_p= t|_p$ and $(s\sigma)|_{p'} = t|_{p'}$.
        This contradicts the linearity of $t$.
    \end{itemize}

    \item By definition, it is clear that $\Pos(s) \subseteq \Pos(s\sigma)$, $\Pos(t) \subseteq \Pos(t\sigma)$, and $\Pos(s\sigma) = \Pos(t\sigma)$.
    It suffices to show that $\Pos(s) \cup \Pos(t) \supseteq \Pos(s\sigma)$.
    We proceed by contradiction.
    Assume that $\Pos(s) \cup \Pos(t) \not\supseteq \Pos(s\sigma)$.
    Then, there exists a position $p$ of $s\sigma$ such that $p \notin \Pos(s) \cup \Pos(t)$.
    Then, there exists an index $i \in \{1,\ldots,n\}$ and a position $p' \in \Pos(s|_{p_i}\sigma)$ such that $p = p_ip'$.
    It follows from the construction of $\sigma$ that $p' \in \Pos(t|_{p_i})$, and thus $p \in \Pos(t)$.
    This contradicts the fact that $p \notin \Pos(s) \cup \Pos(t)$.

    \item The height of a term is the maximum length of positions of the term.
    Therefore, this statement follows~(2).

    \item Since $s,t \in T(\Sigma',\cV)$, it follows from the construction of $\sigma$ that $\Ran(\sigma) \subseteq T(\Sigma',\cV)$.
    Therefore, $\sigma$ is a $\Sigma'$-substitution.
    %
\qed
\end{enumerate}
\end{proof}

The following proposition shows correctness of $\Cosubst[\Sigma']{\cdot}$ for linear $\cC$-substitutions, together with some properties.

\begin{restatable}[correctness of {$\Cosubst[\Sigma']{\sigma}$}]{proposition}{PropCosubstsCorrectness}
\label{prop:Cosubsts-correctness}
Let 
$X \subseteq \cV$, and $\sigma$ be an $X$-linear $\cC$-substitution.
Then, $\{ \rho \mid \rho \in \Cosubst[\cC]{\sigma}, \, t\rho \ne t\sigma \}$ is a finite complement of $\sigma$ w.r.t.\ any linear term $t:\iota$ with $\Var(t) \subseteq X$,
(i.e.,
$\GInst[\cC]{t} = \GInst[\cC]{t\sigma} \uplus \bigcup_{\rho \in \Cosubst[\cC]{\sigma},\, t\rho \ne t\sigma} \GInst[\cC]{t\rho}$)
such that
\begin{enumerate}
\renewcommand{\labelenumi}{(\arabic{enumi})}
    \item any substitution $\rho \in \Cosubst[\cC]{\sigma}$ is an $X$-linear substitution
    such that $\Dom(\rho) = \Dom(\sigma)$, $\rho \not\lesssim \sigma$, and $\sigma \not\lesssim \rho$,
    \item if $\sigma$ is $X$-linearity-preserving, then
    all substitutions in $\Cosubst[\cC]{\sigma}$ are $X$-linearity-preserving,
    \item if $\Ran(\sigma) \subseteq \cV$, then $\Cosubst[\Sigma']{\sigma} = \emptyset$, 
    \item $\Height{x\sigma} \geq \Height{x\rho}$ for any substitution $\rho \in \Cosubst[\cC]{\sigma}$ and variable $x \in \cV$,
    \item for any substitution $\rho \in \Cosubst[\cC]{\sigma}$ and $Y \subseteq \Dom(\sigma)$,
        if $\Ran(\sigma|_Y) \not\subseteq \cV$,
        then there exists at least one variable $x \in \Dom(\rho)\cap Y$ such that $x\rho \notin \cV$ (and thus, $x < x\rho$),
        and
    \item $\{ t\sigma \} \cup \{ t\rho \mid \rho \in \Cosubst[\Sigma']{\sigma} \}$ is a finite pairwise $\cC$-non-overlapping set of linear $\cC$-terms with sort $\iota$.
\end{enumerate}
\end{restatable}
\begin{proof}
Let $\Dom(\sigma) = \{x_1:\iota_1,\ldots,x_n:\iota_n\}$, and $t$ be a linear term with sort $\iota$ such that $\Var(t) \subseteq X$.
Since $\sigma$ is $X$-linear, $x_1\sigma,\ldots,x_n\sigma$ are linear.
It follows from Proposition~\ref{prop:Cocterms-correctness} that for any $i \in \{1,\ldots,n\}$,
\begin{enumerate}
\renewcommand{\labelenumi}{(\roman{enumi})}
\leftskip=1ex
    \item $\Cocterm[\Sigma']{x_i\sigma}$ is a finite pairwise $\cC$-non-overlapping set of linear non-variable $\cC$-terms with sort $\iota_i$,
    \item $x_i\sigma$ and $u$ are $\cC$-non-overlapping for any term $u \in \Cocterm[\cC]{x_i\sigma}$,
        and
    \item $\Height{x_i\sigma}+1 \geq \Height{u}$ for any term $u \in \Cocterm[\cC]{x_i\sigma}$.
\end{enumerate}

We first prove the statements~(1)--(6) in order.
\begin{enumerate}
\renewcommand{\labelenumi}{(\arabic{enumi})}
    \item 
    Let $\rho \in \Cosubst[\cC]{\sigma}$.
    It follows from~(i) that $\rho$ is $X$-linear and $\Dom(\sigma) = \Dom(\rho)$.
    We now show that $\rho \not\lesssim \sigma$ and $\sigma \not\lesssim \rho$.
    We proceed by contradiction.
    Assume that $\rho \lesssim \sigma$ or $\sigma \lesssim \rho$.
    Then, there exists an index $i \in \{1,\ldots,n\}$ such that $x_i\rho \lesssim x_i\sigma$ or $x_i\sigma \lesssim x_i\rho$.
    Then, we have that $\GInst[\cC]{x_i\rho} \supseteq \GInst[\cC]{x_i\sigma}$ or $\GInst[\cC]{x_i\sigma} \supseteq \GInst[\cC]{x_i\rho}$.
    Thus, $x_i\rho$ and $x_i\sigma$ are not $\Sigma'$-non-overlapping.
    This contradicts~(ii).

    \item 
    Assume that $\sigma$ is $X$-linearity-preserving.
    Then, we have that
    $\Var(x_i\sigma)\cap\Var(x_j\sigma) = \emptyset$ for any $i,j \in \{1,\ldots,n\}$ with $i\ne j$.
    Let $\rho \in \Cosubst[\cC]{\sigma}$.
    By definition, we have that $\rho \ne \sigma$.
    We show that $\rho$ is $X$-linearity-preserving.
    We proceed by contradiction.
    Assume that $\rho$ is not $X$-linearity-preserving.
    Then, one of the following hold:
    \begin{enumerate}
    \renewcommand{\labelenumii}{(\roman{enumii})}
    \addtocounter{enumii}{3}
    \leftskip=1ex
        \item $x_i\rho$ is not linear for some $i \in \{1,\ldots,n\}$,
            or
        \item there are two distinct variables $x,x'$ such that $\Var(x\rho) \cap \Var(x'\rho) \ne \emptyset$.
    \end{enumerate}
    Since $\Cocterm[\Sigma']{x_1\sigma},\ldots,\Cocterm[\Sigma']{x_n\sigma}$ are sets of linear terms, all $x_1\rho,\ldots,x_n\rho$ are linear,
    and thus~(iv) does not hold and~(v) holds.
    Let $x,x'$ be distinct variables such that $\Var(x\rho) \cap \Var(x'\rho) \ne \emptyset$.
    Let $y \in \Var(x\rho) \cap \Var(x'\rho)$.
    We make a case analysis depending on whether $y \in \Var(x\sigma)$.
    \begin{itemize}
        \item Case where $y \in \Var(x\sigma)$.
        It follows from the definition of $\Cocterm[\Sigma']{\cdot}$ over $\cC$-terms that $y \in \Var(x'\sigma)$.
        This contradicts the assumption that $\sigma$ is $X$-linearity-preserving.
        \item Case where $y \notin \Var(x\sigma)$.
        Then, $y$ is a variable introduced in constructing $x\rho$, and thus $y$ is fresh, i.e., $y \notin \Var(x'\sigma)$.
        Since $y \in \Var(x'\rho) \setminus \Var(x'\sigma)$, $y$ is freshly introduced in constructing $x'\rho$, and thus $y \notin \Var(x\rho)$.
        This contradicts the fact that $y \in \Var(x\rho) \cap \Var(x'\rho)$.
    \end{itemize}

    \item Assume that $\Ran(\sigma) \subseteq \cV$.
    Then, for any $i\in \{1,\ldots,n\}$, we have that $\Cocterm[\Sigma']{x_i\sigma} = \emptyset$.
    By definition, this implies that $\Cosubst[\Sigma']{\sigma} = \emptyset$.

    \item Let $\rho \in \Cosubst[\Sigma']{\sigma}$ and $x \in \cV$.
    By definition, we have that if $x \notin \Dom(\sigma)$, then $x\sigma=x\rho=x$ and hence $\Height{x\sigma} =  \Height{x\rho} = \Height{x} = 1$.
    Thus, we consider the remaining case where $x \in \Dom(\sigma)$.
    By definition, we have that $x\rho \in \Cocterm[\Sigma']{x\sigma}$.
    It follows from Proposition~\ref{prop:Cocterms-correctness}~(3) that 
    $\Height{x\sigma} \geq \Height{x\rho}$.

    \item Let $\rho \in \Cosubst[\cC]{\sigma}$ and $Y \subseteq \Dom(\sigma)$ such that $\Ran(\sigma|_Y) \not\subseteq \cV$.
    Since $\Ran(\sigma|_Y) \not\subseteq \cV$, there exists a variable $x \in \Dom(\sigma) \cap Y$ such that $x\sigma$ is not a variable.
    Since $\Cocterm[\cC]{x\sigma}$ is a set of terms in $T(\cC,\cV)\setminus \cV$, $x\rho$ is not a variable.
    Therefore, this statement holds.
    
    \item 
    Since $\sigma$ is $X$-linearity-preserving, it follows from Proposition~\ref{prop:linear-subst-properties} that $t\sigma$ is linear.
    It follows from~(1) that $\{t\rho \mid \rho \Cosubst[\cC]{\sigma} \}$ is a finite set of linear $\cC$-terms with sort $\iota$.
    Thus, $\{ t\sigma \} \cup \{ t\rho \mid \rho \in \Cosubst[\cC]{\sigma} \}$ is a finite set of linear $\cC$-terms with sort $\iota$.
    
    It is clear that $\{t\sigma\}$ and $\{t\rho \mid \rho \Cosubst[\Sigma']{\sigma}, \, t\rho \ne t\sigma \}$ are disjoint (i.e., $\{t\sigma\} \cap \{t\rho \mid \rho \in \Cosubst[\Sigma']{\sigma}, \, t\rho \ne t\sigma \} = \emptyset$).
    To complete the proof of this statement, we show 
    that 
    \begin{enumerate}
    \renewcommand{\labelenumii}{(\alph{enumii})}
    \leftskip=1ex
        \item $\{t\sigma\}$ and $\{ t\rho \mid \rho \in \Cosubst[\Sigma']{\sigma}, \, t\rho \ne t\sigma \}$ are $\cC$-non-overlapping,
            and
        \item $\{ t\rho \mid \rho \in \Cosubst[\Sigma']{\sigma}, \, t\rho \ne t\sigma \}$ is $\cC$-non-overlapping.
    \end{enumerate}
    We prove~(a) and~(b) in order.
    \begin{enumerate}
    \renewcommand{\labelenumii}{(\alph{enumii})}
    \leftskip=1ex
    \item 
    We show that $t\sigma$ and $t\rho$ is $\cC$-non-overlapping for any $\rho \in \Cosubst[\cC]{\sigma}$ such that $t\rho \ne t\sigma$.
    We proceed by contradiction.
    Assume that there exists a substitution $\rho \in \Cosubst[\Sigma']{\sigma}$ such that $t\rho \ne t\sigma$ and $t\sigma,t\rho$ are not $\cC$-non-overlapping.
    Then, we have that $\GInst[\cC]{t\sigma} \cap \GInst[\cC]{t\rho} \ne \emptyset$.
    Thus, there exists substitutions $\theta_1,\theta_2$ such that $(t\sigma)\theta_1 = (t\rho)\theta_2 \in \GInst[\cC]{t\sigma} \cap \GInst[\cC]{t\rho}$.
    Since $t\rho \ne t\sigma$ and $(t\sigma)\theta_1 = (t\rho)\theta_2$, we have that $\sigma \lesssim \rho$ or $\rho \lesssim \sigma$.
    It follows from~(1) that $\sigma \not\lesssim \rho$ and $\rho \not\lesssim \sigma$.
    This contradicts the fact that $\sigma \lesssim \rho$ or $\rho \lesssim \sigma$.

    \item
    We show that $t\rho$ and $t\rho'$ are $\cC$-non-overlapping for any $\rho,\rho' \in \Cosubst[\cC]{\sigma}$ with $t\rho\ne t\rho'$.
    Since $t\rho \ne t\rho'$, there exists a variable $x \in \Var(t)$ such that $x\rho \ne x\rho'$.
    By definition, we have that $x\rho,x\rho' \in \Cocterm[\Sigma']{x\sigma}$.
    It follows from~(1) that $x\rho$ and $x\rho'$ are $\cC$-non-overlapping.
    It follows from Proposition~\ref{prop:sufficient-condition-for-F-non-overlappingness-of-terms} that $t\rho$ and $t\rho'$ are $\cC$-non-overlapping.
    \end{enumerate}
%
\end{enumerate}

Finally, we show that
$\{ \rho \mid \rho \in \Cosubst[\cC]{\sigma}, \, t\rho \ne t\sigma \}$ is a finite complement of $\sigma$ w.r.t.\ $t$.
It follows from~(i) and the definition of $\Cosubst[\cC]{\sigma}$ that $\Cosubst[\cC]{\sigma}$ is a set of $\cC$-substitutions $\rho$ such that $\Dom(\rho) = \Dom(\sigma)$.
Since $\Cocterm[\cC]{x_1\sigma},\ldots,\Cocterm[\cC]{x_n\sigma}$ are finite sets of linear $\cC$-terms, $\Cosubst[\Sigma']{\sigma}$ is a finite set of $\cC$-substitutions.
To complete the proof, we show that $\Cosubst[\cC]{\sigma}$ is a complement of $\sigma$ w.r.t.\ $t$, i.e., 
$\GInst[\cC]{t} = \GInst[\cC]{t\sigma} \uplus \bigcup_{\rho \in \Cosubst[\cC]{\sigma},\, t\rho \ne t\sigma} \GInst[\cC]{t\rho}$.
It follows from~(6) that 
$\GInst[\cC]{t\sigma} \cap (\bigcup_{\rho \in \Cosubst[\cC]{\sigma},\, t\rho \ne t\sigma} \GInst[\cC]{t\rho}) = \emptyset$.
By definition, it is trivial that $\GInst[\cC]{t} \supseteq \GInst[\cC]{t\sigma} \uplus \bigcup_{\rho \in \Cosubst[\cC]{\sigma},\, t\rho \ne t\sigma} \GInst[\cC]{t\rho}$.
Thus, it suffices to show that $\GInst[\cC]{t} \subseteq \GInst[\cC]{t\sigma} \cup \bigcup_{\rho \in \Cosubst[\cC]{\sigma},\, t\rho \ne t\sigma} \GInst[\cC]{t\rho}$.
We proceed by contradiction.
Assume that $\GInst[\cC]{t} \not\subseteq \GInst[\cC]{t\sigma} \cup \bigcup_{\rho \in \Cosubst[\cC]{\sigma},\, t\rho \ne t\sigma} \GInst[\cC]{t\rho}$.
Let $\theta$ be a $\Var(t)$-ground $\cC$-substitution such that $t\theta \in \GInst[\cC]{t}$
and
$t\theta \notin \GInst[\cC]{t\sigma} \cup \bigcup_{\rho \in \Cosubst[\cC]{\sigma},\, t\rho \ne t\sigma} \GInst[\cC]{t\rho}$.
Then, we have that $\sigma \not\lesssim \theta$ and $\rho \not\lesssim \theta$ for any $\rho \in \Cosubst[\cC]{\sigma}$.
It follows from $\sigma \not\lesssim \theta$ that
there exists a variable $x \in \Dom(\sigma)$ such that 
$x\sigma \not\leq x\theta$.
Let $p$ be the leftmost outermost position of $x\sigma$ w.r.t.\ $\Root((x\sigma)|_p) \ne \Root((x\theta)|_p)$.
Then, by the definition of $\Cosubst[\cC]{\sigma}$, there exists a substitution $\rho \in \Cosubst[\cC]{\sigma}$ such that $x\rho \leq x\theta$.
This contradicts the fact that $\rho \not\lesssim \theta$.
\qed
\end{proof}

In the rest of the appendix, 
we denote $\{ s\rho \mid \rho \in \Cosubst[\cC]{\sigma}, \, s\rho \ne s\sigma \}$
by $\Copattern[\cC]{s}{\sigma}$.
The following proposition shows the correctness and some properties of $\Copattern[\Sigma']{\cdot}{\cdot}$ for pairs of terms and linear substitutions.

\begin{restatable}[correctness of {$\Copattern[\Sigma']{t}{\sigma}$}]{proposition}{PropCopatternsCorrectness}
\label{prop:Copatterns-correctness}
Let 
$X \subseteq \cV$, $t:\iota$ be a term, and $\sigma$ an $X$-linear $\cC$-substitution such that $\Var(t) \subseteq X$.
Then, $\{ t\rho \mid \rho \in \Copattern[\Sigma']{t}{\sigma} \}$ is a finite pairwise $\cC$-non-overlapping complement of $t\sigma$
(i.e., $\GInst[\Sigma']{t} = \GInst[\cC]{ t\sigma } \uplus \GInst[\cC]{\Copattern[\cC]{t}{\sigma}}$)
such that 
\begin{enumerate}
\renewcommand{\labelenumi}{(\arabic{enumi})}
    \item 
    any term $u \in \Copattern[\cC]{t}{\sigma}$ is a non-variable term with sort $\iota$,
    \item if $\Ran(\sigma|_{\Var(t)}) \subseteq \cV$, then $\Copattern[\cC]{t}{\sigma} = \emptyset$,
    \item if $t$ is a pattern, then all terms in $\Copattern[\cC]{t}{\sigma}$ are patterns,
    \item if $\sigma$ is $X$-linearity-preserving, then
    any term $u \in \Copattern[\cC]{t}{\sigma}$ is linear,
    \item $t < u$ for any term $u \in \Copattern[\cC]{t}{\sigma}$,
        and
    \item $\Height{t\sigma} \geq \Height{u}$ for any term $u \in \Copattern[\cC]{t}{\sigma}$.
\end{enumerate}
\end{restatable}
\begin{proof}
It follows from Proposition~\ref{prop:Cocterms-correctness} that
\begin{enumerate}
\renewcommand{\labelenumi}{(\alph{enumi})}
\leftskip=1ex
    \item $\Cosubst[\Sigma']{\sigma}$ is a finite set of $X$-linear $\cC$-substitutions $\rho$ such that $\Dom(\rho) = \Dom(\sigma)$, $\rho \not\lesssim \sigma$, and $\sigma \not\lesssim \rho$,
    \item if $\sigma$ is $X$-linearity-preserving, then any $\rho \in \Cosubst[\cC]{\sigma}$ is $X$-lineaerity-preserving, 
    \item if $\Ran(\sigma) \subseteq \cV$, then $\Cosubst[\Sigma']{\sigma} = \emptyset$, 
    \item $\Height{x\sigma} \geq \Height{x\rho}$ for any substitution $\rho \in \Cosubst[\Sigma']{\sigma}$ and variable $x \in \cV$,
    \item for any substitution $\rho \in \Cosubst[\Sigma']{\sigma}$ and $Y \subseteq \Dom(\sigma)$,
        if $\Ran(\sigma|_Y) \not\subseteq \cV$,
        then there exists at least one variable $x \in \Dom(\rho)\cap Y$ such that $x\rho \notin \cV$,
        and
    \item $\{ t\rho \mid \rho \in \Cosubst[\Sigma']{\sigma} \}$ is a finite pairwise $\cC$-non-overlapping complement of $t\sigma$
        (i.e., $\GInst[\cC]{t} = \GInst[\cC]{\{ t\sigma \}} \uplus \GInst[\cC]{\{ t\rho \mid \rho \in \Cosubst[\cC]{\sigma} \}}$).
\end{enumerate}
It follows from~(f) that
$\{ t\rho \mid \rho \in \Copattern[\Sigma']{t}{\sigma} \}$ is a finite pairwise $\cC$-non-overlapping complement of $t\sigma$.
The first statement~(1) follows~(a).
The second statement~(2) follows Proposition~\ref{prop:Cosubsts-correctness}~(3).
The fourth statement~(4) follows~(d).
We show the remaining statements~(3),\,(5),\,(6) in order.
\begin{enumerate}
\renewcommand{\labelenumi}{(\arabic{enumi})}
\addtocounter{enumi}{2}
\item
Assume that $t$ is a pattern.
It follows from~(a) that all substitutions in $\Cosubst[\cC]{\sigma}$ are $\cC$-substitutions.
Thus, for any $\rho \in \Cosubst[\cC]{\sigma}$, $t\rho$ is a pattern.
Therefore, all terms in $\Copattern[\cC]{t}{\sigma}$ are patterns.

\stepcounter{enumi}

\item
Let $u \in \Copattern[\Sigma']{t}{\sigma}$.
We proceed by contradiction.
Assume that $t \not< u$.
By definition, there exists a substitution $\rho \in \Cosubst[\cC]{\sigma}$ such that $u=t\rho$, and thus $t \lesssim t\rho = u$.
By the assumption, we have that $t = u$, i.e., $\rho$ is the identity substitution.
Since $\rho \in \Cosubst[\cC]{\sigma}$, we have that $\Cosubst[\cC]{\sigma} \ne \emptyset$ and $\Ran(\sigma|_{\Var(t)}) \not\subseteq \cV$.
Thus, it follows from~(e) above that there exists a variable $x \in \Dom(\rho) \cap \Var(t)$ such that $x\rho \notin \cV$.
Hence, we have that $t \ne t\rho$.
This contradicts the fact that $t = u$.

\item
Let $u \in \Copattern[\cC]{t}{\sigma}$ and $\rho \in \Cosubst[\cC]{\sigma}$ such that $u=t\rho$.
We prove this statement by structural induction on $t$.
Since the case where $t \in \cV$ is trivial by~(d), we consider the remaining case where $t \notin \cV$.
Let $t = f(t_1,\ldots,t_n)$.
Then, we have that $\Height{t\sigma} = 1 + \max\{ \Height{t_i\sigma} \mid 1 \leq i \leq n \}$.
By the induction hypothesis, we have that $\Height{t_i\sigma} \geq \Height{t_i\rho}$ for any $i \in \{1,\ldots,n\}$.
Therefore, we have that $\Height{t\sigma} = 1 + \max\{ \Height{t_i\sigma} \mid 1 \leq i \leq n \} \geq 1 + \max\{ \Height{t_i\rho} \mid 1 \leq i \leq n \} = \Height{t\rho} = \Height{u}$.
\qed
\end{enumerate}
\end{proof}

The following proposition shows the correctness of $\ominus$ over unconstrained patterns.
\begin{restatable}[correctness of $\ominus$ over patterns]{proposition}{PropPatternOminusProperties}
\label{prop:pattern-ominus-properties}
Let 
$s:\iota$ be a pattern, and $t:\iota'$ a linear pattern. 
Then, 
$s \ominus t$ is a finite pairwise $\cC$-non-overlapping set of patterns with sort $\iota$
such that 
\begin{enumerate}
\renewcommand{\labelenumi}{(\arabic{enumi})}
    \item if $s$ is linear, then $s \ominus t$ is a set of linear patterns,
    \item $s < u$ for any term $u \in {s \ominus t}$,
    \item $\max\{\Height{s},\Height{t}\} \geq \Height{u}$ for any term $u \in {s \ominus t}$,
        and
    \item $\GInst[\cC]{s \ominus t} = \GInst[\cC]{s} \setminus \GInst[\cC]{t}$.
\end{enumerate}
\end{restatable}
\begin{proof}
Since the case where $s \ominus t = \{s\}$ is trivial, we consider the remaining case where $s \ominus t \ne \{s\}$.
Let $s \ominus t = \Copattern[\cC]{s}{\sigma}$, where $t'$ is a renamed variant of $t$ with $\Var(s)\cap\Var(t')=\emptyset$ and $\sigma = \mgu[\cC](s,t')$.
It follows from Proposition~\ref{prop:Copatterns-correctness} that
\begin{enumerate}
\renewcommand{\labelenumi}{(\alph{enumi})}
\leftskip=1ex
    \item $\Copattern[\Sigma']{s}{\sigma}$ is a finite pairwise $\cC$-non-overlapping set of linear patterns with sort $\iota$,
    \item if $\sigma$ is $\Var(s,t')$-linearity-preserving, then any $u \in s \ominus t$ is linear,
    \item if $\Ran(\sigma|_{\Var(s)}) \subseteq \cV$, then $\Copattern[\Sigma']{s}{\sigma} = \emptyset$,
    \item $s < u$ for any term $u \in \Copattern[\Sigma']{s}{\sigma}$,
    \item $\Height{s\sigma} \geq \Height{u}$ for any term $u \in \Copattern[\Sigma']{s}{\sigma}$,
        and
    \item $\GInst[\cC]{s} = \GInst[\cC]{s\sigma} \uplus \GInst[\cC]{\Copattern[\cC]{s}{\sigma}}$.
\end{enumerate}
Thus, it follows from~(a) that
$s \ominus t$ is a finite pairwise $\cC$-non-overlapping set of patterns with sort $\iota$.
The first statement~(1) follows~(b).
The second statement~(2) follows~(d).
It follows from Proposition~\ref{prop:propeties-of-F-unifiable-linear-terms}~(3) and~{e)} that the third statement~(3) holds.

We show the fourth statement~(4).
Since $\sigma$ is an mgu of $s,t'$, we have that $s\sigma = t'\sigma$.
Since $t'$ is a renamed variant of $t$, we have that $\GInst[\cC]{t} = \GInst[\cC]{t'}$ and $\GInst[\cC]{t} \supseteq \GInst[\cC]{t'\sigma}$.
It follows from~(f) that $\GInst[\cC]{\Copattern[\cC]{s}{\sigma}} = \GInst[\cC]{s} \setminus \GInst[\cC]{s\sigma}$.
It follows from Proposition~\ref{prop:disjointness-of-F-unifiable-terms} that
$\GInst[\cC]{s\sigma}$ ($=\GInst[\cC]{t\sigma}$), $\GInst[\cC]{s} \setminus \GInst[\cC]{s\sigma}$, and  $\GInst[\cC]{t'} \setminus \GInst[\cC]{t'\sigma}$ are pairwise $\cC$-non-overlapping.
Thus, we have that $\GInst[\cC]{s} \setminus \GInst[\cC]{s\sigma} = \GInst[\cC]{s} \setminus \GInst[\cC]{t'}$.
Therefore, we have that $\GInst[\cC]{s \ominus t} = \GInst[\cC]{s} = \GInst[\cC]{s\sigma} \uplus \GInst[\cC]{\Copattern[\cC]{s}{\sigma}} = \GInst[\cC]{s} \setminus \GInst[\cC]{t'} = \GInst[\cC]{s} \setminus \GInst[\cC]{t}$.
\qed
\end{proof}


\ThmConstrainedPatternOminusPropertiesX*
\begin{proof}[Sketch]
This is a straightforward extension of Proposition~\ref{prop:pattern-ominus-properties} to constrained patterns by means of the following property:
\begin{itemize}
    \item $\Copattern{s}{\sigma}$ is a set of value-free patterns because $s,t$ are value-free,
    \item $\sigma|_{\Var(\phi,\psi')} = \rho|_{\Var(\phi,\psi')}$ for any $\rho \in \Cosubst[\cC]{\sigma}$ because $s,t$ are value-free and there is no constructor $c$ in $\cC \setminus \Val$ such that
    $c: \iota_1 \times \cdots \times \iota_n \to \iota \in \cStheory$,
        and
    \item
$
    \GInst[\cC]{\CTerm{s\sigma}{\phi\sigma}} = 
    \GInst[\cC]{\CTerm{s\sigma}{\phi\sigma \land \psi'\sigma}} \uplus 
    \GInst[\cC]{\CTerm{s\sigma}{\phi\sigma \land \neg\psi'\sigma}}
$ for an arbitrary constraint $\psi'$.
\qed
\end{itemize}
\end{proof}

\section{Proof of Theorem~\ref{thm:constrained-set-ominus-properties}}

It is not so trivial to prove correctness of $\ominus$ over sets of unconstrained linear patterns.
For this reason, we first show correctness of the unconstrained case.



To prove the correctness of $\ominus$ over sets of linear patterns, we need induction with a well-founded order.
Let us consider the case where $P$ is pairwise $\cC$-non-overlapping, $P = P' \uplus \{s\}$, $Q = Q' \uplus \{t\}$, and
$s \ominus t \ne \{s\}$.
In this case, to compute $P \ominus Q$, we recursively call $\ominus$ for 
sets of linear patterns 
as $(P' \cup (s \ominus t)) \ominus (Q' \cup (t \ominus s))$.
It follows from Proposition~\ref{prop:pattern-ominus-properties}~(3) that the maximum height of terms in $P\cup Q$ is not less than that of 
$P' \cup (s \ominus t) \cup Q' \cup (t \ominus s)$, i.e., 
$\max\{ \Height{u} \mid u \in P\cup Q \} \geq \max\{ \Height{u'} \mid u' \in P' \cup (s \ominus t) \cup Q' \cup (t \ominus s) \}$.
Let $h = \max\{ \Height{u} \mid u \in P\cup Q \}$.
Then, in computing $P \ominus Q$, the heights of all terms considered during the computation are less than or equal to $h$, and thus such terms are finitely many.
We denote by $T(\Sigma,\cV)_{\leq h}$ the set of terms whose heights are less than or equal to $h$.
Then, we define a quasi-order $\SuccsimH$ over terms in $T(\Sigma,\cV)$ as follows:
$s \SuccsimH t$ if and only if $s \lesssim t$, $\Height{s} \leq h$, and $\Height{t} \leq h$.
We define the strict part $\SuccH$ of $\SuccsimH$ as ${\SuccH} = ({\SuccsimH} \setminus {\PrecsimH})$.
Note that $s \SuccH t$ if and only if $s < t$, $\Height{s} \leq h$, and $\Height{t} \leq h$.

\begin{restatable}{proposition}{PropWellFoundednessOfSuccH}
\label{prop:well-foundedness-of-SuccH}
For any $h \in \mathbb{N}$, $\SuccH$ is a well-founded order.
\end{restatable}
\begin{proof}
We proceed by contradiction.
Assume that $\SuccH$ is not well-founded.
Then, there exists an infinite sequence $s_1 \SuccH s_2 \SuccH s_3 \SuccH \cdots$.
Thus, we have that $s_1 < s_2 < s_3 < \cdots$ and $\Height{s_i} \leq h$ for all $i \geq 1$.
Thus, we have that $\{ s_1, s_2, s_3, \ldots \}$ is infinite up to renaming and $\{ s_1, s_2, s_3, \ldots \} \subseteq T(\Sigma,\cV)_{\leq h}$.
Since $\Sigma$ is finite, $T(\Sigma,\cV)_{\leq h}$ is finite up to renaming, and thus $\{ s_1, s_2, s_3, \ldots \}$ is finite up to renaming.
This contradicts the fact that $\{ s_1, s_2, s_3, \ldots \}$ is infinite up to renaming.
\qed
\end{proof}
We denote the multiset extensions of $\SuccsimH$ and $\SuccH$ by $\SuccsimH^M$ and $\SuccH^M$, respectively.
The following proposition follows Proposition~\ref{prop:well-foundedness-of-SuccH} and~\cite[Lemma~2.5.4]{BN98}.
\begin{restatable}{proposition}{PropWellFoundednessOfSuccHM}
\label{prop:well-foundedness-of-SuccHM}
For any $h \in \mathbb{N}$, $\SuccH^M$ is a well-founded order over finite sets.
\end{restatable}

The following proposition shows the correctness of $\ominus$ over sets of linear patterns.

\begin{restatable}[correctness of $\ominus$ over sets of linear patterns]{proposition}{PropSetOminusProperties}
\label{prop:set-ominus-properties}
Let 
$P,Q$ be finite sets of linear patterns.
Then, all of the following statements hold:
\begin{enumerate}
\renewcommand{\labelenumi}{(\arabic{enumi})}
    \item if $P$ is pairwise $\cC$-non-overlapping, $P = P'\uplus \{s\}$, $Q=Q' \uplus \{t\}$, and
    $s \ominus t \ne \{s\}$,
    then 
        $P'$ and $s \ominus t$ are $\cC$-non-overlapping,
    \item if there are no patterns $s \in P$ and $t \in Q$ such that
    $s \ominus t \ne \{s\}$,
    then $\GInst[\cC]{P \ominus Q} = \GInst[\cC]{P}$,
        and
    \item if $P$ is pairwise $\cC$-non-overlapping, then $\GInst[\cC]{P \ominus Q} = \GInst[\cC]{P} \setminus \GInst[\cC]{Q}$.
\end{enumerate}
\end{restatable}
\begin{proof}
We prove the statements in order.
\begin{enumerate}
\renewcommand{\labelenumi}{(\arabic{enumi})}
    \item Suppose that $P$ is pairwise $\cC$-non-overlapping.
    Let $P = P'\uplus \{s\}$, $Q=Q' \uplus \{t\}$ such that 
    $s \ominus t \ne \{s\}$.
    It suffices to show that $P'$ and $\{u\}$ are $\cC$-non-overlapping for any $u \in (s \ominus t)$.
    By definition, $u$ is a $\cC$-instance of $s$, and thus $\GInst[\cC]{u} \subseteq \GInst[\cC]{s}$.
    Since $P$ is $\cC$-non-overlapping, we have that $\GInst[\cC]{P'} \cap \GInst[\cC]{s} = \emptyset$, and thus $\GInst[\cC]{P'} \cap \GInst[\cC]{u} = \emptyset$.
    Therefore, $P'$ and $\{u\}$ are $\cC$-non-overlapping.

    \item Trivial by definition.

    \item Suppose that $P$ is pairwise $\cC$-non-overlapping.
    Let $h \geq \max\{ \Height{u} \mid  u \in P \cup Q\}$.
    We prove this statement by well-founded induction on $P$ with $\SuccH^M$.
    We make a case analysis depending on whether there are patterns $s \in P$ and $t \in Q$ such that 
    $s \ominus t \ne \{s\}$.
    \begin{itemize}
        \item Case where there are no patterns $s \in P$ and $t \in Q$ such that 
        $s \ominus t \ne \{s\}$.
        Then, by definition, we have that $P \ominus Q = P$ and hence $\GInst[\cC]{P \ominus Q} = \GInst[\cC]{P}$.
        We proceed by contradiction.
        Assume that $\GInst[\cC]{P} \ne \GInst[\cC]{P} \setminus \GInst[\cC]{Q}$.
        It is clear that $\GInst[\cC]{P} \supseteq \GInst[\cC]{P} \setminus \GInst[\cC]{Q}$.
        Thus, by assumption, we have that $\GInst[\cC]{P} \not\subseteq \GInst[\cC]{P} \setminus \GInst[\cC]{Q}$.
        Let $u \in \GInst[\cC]{P} \cap \GInst[\cC]{Q}$.
        Then, there exists terms $s \in P$ and $t \in Q$ such that $u \in \GInst[\cC]{s} \cap \GInst[\cC]{t}$.
        Thus, $s,t$ are not $\cC$-non-overlapping,
        i.e., 
        $s,t'$ are unifiable for some renamed variant $t'$ of $t$ with $\Var(s)\cap\Var(t')=\emptyset$, and thus $s \ominus t \ne \{s\}$.
        This contradicts the fact that there are no patterns $s \in P$ and $t \in Q$ such that
        $s \ominus t \ne \{s\}$.

        \item Case where there are patterns $s \in P$ and $t \in Q$ such that 
        $s \ominus t \ne \{s\}$.
        Let $P = P' \uplus \{s\}$ and $Q = Q' \uplus \{t\}$ such that
        $s \ominus t \ne \{s\}$.
        Then, we have that $P \ominus Q = (P' \uplus (s \ominus t)) \ominus (Q' \cup (t \ominus s))$.
        It follows from Proposition~\ref{prop:pattern-ominus-properties} that 
        $\{s\} \SuccH^M (s \ominus t)$ and hence $(P' \uplus \{s\}) \SuccH^M (P' \uplus (s \ominus t))$.
        By the induction hypothesis, we have that
        $\GInst[\cC]{(P' \uplus (s \ominus t)) \ominus (Q' \cup (t \ominus s))} = 
        \GInst[\cC]{P' \uplus (s \ominus t)} \setminus \GInst[\cC]{Q' \cup (t \ominus s)}$.
        It follows from~(1) that $P'$ and $s \ominus t$ are $\cC$-non-overlapping and $Q'$ and $t \ominus s$ are $\cC$-non-overlapping.
        Thus, by definition, we have that $\GInst[\cC]{P' \uplus (s \ominus t)} = \GInst[\cC]{P'} \uplus \GInst[\cC]{s \ominus t}$
        and $\GInst[\cC]{Q' \uplus (t \ominus s)} = \GInst[\cC]{Q'} \uplus \GInst[\cC]{t \ominus s}$.
        It follows from Proposition~\ref{prop:pattern-ominus-properties}~(4) that
        $\GInst[\cC]{s \ominus t} = \GInst[\cC]{s} \setminus \GInst[\cC]{t}$
        and
        $\GInst[\cC]{t \ominus s} = \GInst[\cC]{t} \setminus \GInst[\cC]{s}$.
        Thus, we have that 
        $\GInst[\cC]{(P' \uplus (s \ominus t)) \ominus (Q' \cup (t \ominus s))} 
        = 
        (\GInst[\cC]{P'} \uplus \GInst[\cC]{s \ominus t}) \setminus (\GInst[\cC]{Q'} \uplus \GInst[\cC]{t \ominus s})
        =
        (\GInst[\cC]{P'} \uplus (\GInst[\cC]{s} \setminus \GInst[\cC]{t})) \setminus (\GInst[\cC]{Q'} \uplus (\GInst[\cC]{t} \setminus \GInst[\cC]{s}))
        =
        (\GInst[\cC]{P'} \uplus \GInst[\cC]{s}) \setminus (\GInst[\cC]{Q'} \cup \GInst[\cC]{t})
        $.
        Therefore, we have that 
        $\GInst[\cC]{P \ominus Q} 
        =
        \GInst[\cC]{(P' \uplus (s \ominus t)) \ominus (Q' \cup (t \ominus s))} 
        =
        (\GInst[\cC]{P'} \uplus \GInst[\cC]{s}) \setminus (\GInst[\cC]{Q'} \cup \GInst[\cC]{t})
        =
        \GInst[\cC]{P} \setminus \GInst[\cC]{Q}
        $.
\qed
    \end{itemize}
\end{enumerate}
\end{proof}

Next, we consider the case of constrained patterns.

To prove the correctness of $\ominus$ over sets of unconstrained linear patterns, we prepared the well-founded order $\SuccsimH$.
Unfortunately, this does not work for the correctness of $\ominus$ over sets of constrained linear patterns.
In computing $\CTerm{s}{\phi} \ominus \CTerm{t}{\psi}$, the resulting set may contain $\CTerm{s\sigma}{\phi\sigma \land \neg\psi'\sigma}$ and we may have that $s = s\sigma$ (and thus $(\phi\sigma \land \neg\psi'\sigma) = (\phi \land \neg\psi')$.
On the other hand, since $\phi \land \psi'$ and $\neg\psi'$ are satisfiable, we have that
$\{ \theta \mid \Eval{\phi\theta} = \top \} \supset \{ \theta \mid \Eval{(\phi \land \neg\psi')\theta} = \top \}$.
However, the sets may be infinite, and thus the relation $\supset$ for the sets is not well-founded in general.
To overcome the problem, we consider the number of constrained terms $\CTerm{t}{\psi}$ 
such that $\CTerm{s}{\phi} \ominus \CTerm{t}{\psi} \ne \{ \CTerm{s}{\phi}\}$.
Let $\SuccsimHN$ be the order of lexicographic products of terms and natural numbers compared by $\SuccsimH$ and $\geq_{\mathbb{N}}$.
Let ${\SuccHN} = (\SuccsimHN \setminus \PrecsimHN)$.
Then, it is clear that $\SuccHN$ is well-founded.
We define the weight $w$ for pairs of sets of constrained linear patterns as follows:
\[
w(P,Q) = \{ (s,n) \mid \CTerm{s}{\phi} \in P, \, n = |\{ \CTerm{t}{\psi} \in Q \mid 
\CTerm{s}{\phi} \ominus \CTerm{t}{\phi} \ne \{\CTerm{s}{\phi}\}
\}| \,\}
\]
For the recursive call of $\ominus$ with $P \ominus Q = P' \ominus Q'$, we have that $w(P,Q) \SuccHN w(P',Q')$.
The following proposition shows the correctness of $\ominus$ over sets of constrained linear patterns.

\begin{restatable}{proposition}{PropConstrainedSetOminusPropertiesX}
\label{prop:constrained-set-ominus-propertiesX}
Let 
$P,Q$ be finite sets of value-free constrained linear patterns.
Then, 
all of the following hold:
\begin{enumerate}
\renewcommand{\labelenumi}{(\arabic{enumi})}
    \item if $P$ is pairwise $\cC$-non-overlapping, $P = P'\uplus \{\CTerm{s}{\phi}\}$, $Q=Q' \uplus \{\CTerm{t}{\psi}\}$, and
    $\CTerm{s}{\phi} \ominus \CTerm{t}{\psi} \ne \{\CTerm{s}{\phi}\}$,
    then 
        $P'$ and $(\CTerm{s}{\phi} \ominus \CTerm{t}{\psi})$ are $\cC$-non-overlapping,
    \item if there are no constrained patterns $\CTerm{s}{\phi} \in P$ and $\CTerm{t}{\psi} \in Q$ such that 
    $\CTerm{s}{\phi} \ominus \CTerm{t}{\psi} \ne \{\CTerm{s}{\phi}\}$,
    then $\GInst[\cC]{P \ominus Q} = \GInst[\cC]{P}$,
        and
    \item if $P$ is pairwise $\cC$-non-overlapping, then $\GInst[\cC]{P \ominus Q} = \GInst[\cC]{P} \setminus \GInst[\cC]{Q}$.
\end{enumerate}
\end{restatable}
\begin{proof}
The first two statements are straightforward extensions of the corresponding statements in Proposition~\ref{prop:set-ominus-properties} to constrained linear patterns.
We only show the third statement.
We first show that if $P \ominus Q = P' \ominus Q'$, then $w(P,Q) \SuccHN w(P',Q')$.
Let $P = P'' \uplus \{\CTerm{s}{\phi}\}$ and
$Q = Q'' \uplus \{\CTerm{t}{\psi}\}$ such that 
$\CTerm{s}{\phi} \ominus \CTerm{t}{\psi} \ne \{\CTerm{s}{\phi}\}$.
By definition, we let
$P' = (P'' \cup (\CTerm{s}{\phi} \ominus \CTerm{t}{\psi}))
=
\{\CTerm{s'}{\phi\sigma}\mid s'\in\Copattern[\cC]{s}{\sigma}\}\cup\{\CTerm{s\sigma}{\phi\sigma\land\lnot\psi\sigma}\mid\mbox{$(\phi\sigma\land\lnot\psi\sigma)$ is satisfiable}\}$
and 
$Q' = (Q'' \cup (\CTerm{t}{\psi} \ominus \CTerm{s}{\phi}))
=
\{\CTerm{t'}{\psi\sigma}\mid t'\in\Copattern[\cC]{s}{\sigma}\}\cup\{\CTerm{t\sigma}{\psi\sigma\land\lnot\phi\sigma}\mid\mbox{$(\psi\sigma\land\lnot\phi\sigma)$ is satisfiable}\}$.
Since $P',Q'$ are $\cC$-non-overlapping, we have that 
$w(P,Q) \SuccHN w(P',Q')$.
Then, as for Proposition~\ref{prop:set-ominus-properties}~(3), the third statement can be proved by induction on $\SuccHN$.
\qed
\end{proof}

Theorem~\ref{thm:constrained-set-ominus-properties} is an immediate consequence of Proposition~\ref{prop:constrained-set-ominus-propertiesX}.




\section{Proof of Theorem~\ref{thm:decidability-of-quasi-reducibility-of-LL-LCTRSs}}

\ThmDecidabilityOfQuasiReducibilityOfLLLCTRSs*

Let $\cR$ be a finite left-linear LCTRS 
such that $\Sigterm$ is finite and there is no constructor $c: \iota_1 \times \cdots \times \to \iota \in \cC_\cR$ with $\iota \in \cStheory$.
Then, $\cR$ is quasi-reducible if and only if $\CopatternSet{\{ \CTerm{\ell}{\phi} \mid \ell \to r ~[\phi] \in \cR, ~ \mbox{$\ell$ is a pattern}\}} = \emptyset$.
Thus, quasi-reducibility is decidable for such LCTRSs.

\begin{proof}
Let $Q= \{ \CTerm{\ell}{\phi} \mid \ell \to r ~[\phi] \in \cR, ~ \mbox{$\ell$ is a pattern}\}$.

We first show that 
if $\cR$ is quasi-reducible, then 
$\CopatternSet{Q} = \emptyset$.
We proceed by contradiction.
Assume that $\cR$ is quasi-reducible
and
$\CopatternSet{Q} \ne \emptyset$.
Then, there exists a constrained pattern $\CTerm{s}{\phi} \in \CopatternSet[\cC_\cR]{Q}$.
It follows from Corollary~\ref{cor:correctness-of-CopatternSetX} that
$\GInst[\cC_\cR]{Q} = \{ f(t_1,\ldots,t_n) \mid f \in \cD, ~ t_1,\ldots,t_n \in T(\cC_\cR) \}
\setminus \GInst[\cC_\cR]{Q}$, and thus
$\GInst[\cC_\cR]{\CTerm{s}{\phi}} \cap \GInst[\cC_\cR]{Q} = \emptyset$.
Thus, $\GInst[\cC_\cR]{\CTerm{s}{\phi}} \cap \NF_{\cR} = \emptyset$ and hence there exists an irreducible ground pattern.
This contradicts the assumption that $\cR$ is quasi-reducible.

Next, we show that 
if $\CopatternSet{Q} = \emptyset$, then $\cR$ is quasi-reducible.
It follows from Corollary~\ref{cor:correctness-of-CopatternSetX} that
$\GInst[\cC_\cR]{Q} = \{ f(t_1,\ldots,t_n) \mid f \in \cD, ~ t_1,\ldots,t_n \in T(\cC_\cR) \}
\setminus \GInst[\cC_\cR]{Q}$, and thus
$\GInst[\cC_\cR]{Q} \cap \NF_{\cR} = \emptyset$.
Hence, there is no ground irreducible pattern.
Therefore, $\cR$ is quasi-reducible.

Finally, we show that quasi-reducibility is decidable for left-linear LCTRSs 
such that $\Sigterm$ is finite and there is no constructor $c: \iota_1 \times \cdots \times \to \iota \in \cC$ with $\iota \in \cStheory$.
It suffices to show that $\CopatternSet{Q}$ is computable.
Since $\Sigterm$ is finite and built-in theories are decidable, 
$\Cocterm{\Cdot}$ over $\cC$-terms is computable and satisfiability of constraints is decidable.
Thus, $\ominus$ over constrained linear patterns is computable.
By the well-founded order $\SuccHN$, in computing $\ominus$ over sets of constrained linear patterns, the recursive call of $\ominus$ is terminating.
Therefore, $\CopatternSet{Q}$ is computable. 
\qed
\end{proof}